\documentclass[10pt,journal,compsoc]{IEEEtran}
\usepackage{amsmath}
\usepackage{flushend}
\usepackage{capt-of}
\usepackage{amsthm}
\usepackage{amssymb}
\usepackage{amsfonts}
\usepackage{dsfont}
\usepackage{epsfig}
\usepackage{graphicx}
\usepackage{subfigure}
\usepackage{url}
\usepackage[bottom]{footmisc}
\usepackage{balance}
\usepackage{algorithmic}
\usepackage{algorithm}
\usepackage{multirow}
\usepackage{xspace}
\usepackage{color}
\usepackage{cite}
\usepackage{booktabs}
\usepackage{varwidth}
\usepackage[singlelinecheck=off]{caption}
\DeclareCaptionType{copyrightbox}
\newtheorem{theorem}{Theorem}

\newtheorem{lemma}{Lemma}

\newtheorem{example}{Example}
\newtheorem{problem}{Problem}

\newtheorem{fact}{Fact}
\newcommand{\spara}[1]{\smallskip\noindent{\bf #1}}

\newcommand{\squishlist}{
 \begin{list}{$\bullet$}
  {  \setlength{\itemsep}{0pt}
     \setlength{\parsep}{3pt}
     \setlength{\topsep}{3pt}
     \setlength{\partopsep}{0pt}
     \setlength{\leftmargin}{2em}
     \setlength{\labelwidth}{1.5em}
     \setlength{\labelsep}{0.5em}
} }
\newcommand{\squishlisttight}{
 \begin{list}{$\bullet$}
  { \setlength{\itemsep}{0pt}
    \setlength{\parsep}{0pt}
    \setlength{\topsep}{0pt}
    \setlength{\partopsep}{0pt}
    \setlength{\leftmargin}{2em}
    \setlength{\labelwidth}{1.5em}
    \setlength{\labelsep}{0.5em}
} }

\newcommand{\squishdesc}{
 \begin{list}{}
  {  \setlength{\itemsep}{0pt}
     \setlength{\parsep}{3pt}
     \setlength{\topsep}{3pt}
     \setlength{\partopsep}{0pt}
     \setlength{\leftmargin}{1em}
     \setlength{\labelwidth}{1.5em}
     \setlength{\labelsep}{0.5em}
} }

\newcommand{\squishend}{
  \end{list}
}

\newcommand{\eat}[1]{}

\newcommand{\NP}{\ensuremath{\mathbf{NP}}\xspace}
\newcommand{\PTAS}{\ensuremath{\mathbf{PTAS}}\xspace}
\newcommand{\sharpP}{\ensuremath{\mathbf{\#P}}\xspace}

\newcommand{\Poly}{\ensuremath{\mathbf{P}}\xspace}

\newcounter{ccc}

\DeclareMathOperator*{\argmin}{arg\,min}
\DeclareMathOperator*{\argmax}{arg\,max}
\newcommand{\bigO}{\mathcal{O}}

\newcommand{\maxk}{\mbox{\textsc{max $k$-cover}}\xspace}
\newcommand{\topcat}{\mbox{\textsc{$s$-$t$ top-$k$ catalysts}}\xspace}
\newcommand{\setcov}{\mbox{\textsc{set cover}}\xspace}
\newcommand{\bone}{\mbox{\textsf{Individual top-$k$}}\xspace}
\newcommand{\btwo}{\mbox{\textsf{Greedy}}\xspace}
\newcommand{\ours}{\mbox{\textsf{Rel-Path}}\xspace}
\newcommand{\ipi}{\mbox{\textsc{iterative path inclusion}}\xspace}
\newcommand{\steptwo}{\mbox{\textsf{Iterative Path Inclusion}}\xspace}
\newcommand{\topcatavg}{\mbox{\textsc{top-$k$ catalysts avg}}\xspace}
\newcommand{\topcatmax}{\mbox{\textsc{top-$k$ catalysts max}}\xspace}
\newcommand{\topcatmin}{\mbox{\textsc{top-$k$ catalysts min}}\xspace}
\newcommand{\mincat}{\mbox{\textsc{minimum set of catalysts}}\xspace}

\newcommand{\topcatconn}{\mbox{\textsc{connectivity top-$k$ catalysts}}\xspace}
\newcommand{\boneshort}{\mbox{\textsf{Ind-$k$}}\xspace}
\newcommand{\btwoshort}{\btwo}
\newcommand{\oursshort}{\mbox{\textsf{Rel-Path}}\xspace}

\usepackage{lipsum}

\usepackage{graphicx,wrapfig,lipsum}

\makeatletter
\long\def\@IEEEtitleabstractindextextbox#1{\parbox{0.922\textwidth}{#1}}
\makeatother

\begin{document}

\title{Conditional Reliability in Uncertain Graphs}

\author{Arjiit~Khan,
        Francesco~Bonchi,
        Francesco~Gullo,
        and~Andreas~Nufer
\IEEEcompsocitemizethanks{\IEEEcompsocthanksitem A. Khan is with Nanyang Technological University, Singapore.\protect
\IEEEcompsocthanksitem F. Bonchi is with ISI Foundation, Italy.
\IEEEcompsocthanksitem F. Gullo is with UniCredit, R\&D Dept., Italy.
\IEEEcompsocthanksitem A. Nufer is with ETH Zurich, Switzerland.}
\thanks{The research is supported by MOE Tier-1 RG83/16 and NTU M4081678.}
\thanks{Manuscript received January 31, 2017.}}
\vspace{-5mm}

\vspace{-2mm}
\IEEEtitleabstractindextext{%
\begin{abstract}
Network reliability is a well-studied problem that requires to measure the probability that a target node is reachable from a source node in a
probabilistic (or uncertain) graph, i.e., a graph where every edge is assigned a probability of existence. Many approaches
and problem variants have been considered in the literature, majority of them assuming that edge-existence probabilities are fixed. Nevertheless,
in real-world graphs, edge probabilities typically depend on external conditions. In metabolic networks, a protein can be converted
into another protein with some probability depending on the presence of certain enzymes. In social influence networks, the probability
that a tweet of some user will be re-tweeted by her followers depends on whether the tweet contains specific hashtags. In transportation
networks, the probability that a network segment will work properly or not, might depend on external conditions such as weather or time of the day.

In this paper, we overcome this limitation and focus on \emph{conditional reliability}, that is, assessing reliability when edge-existence
probabilities depend on a set of conditions. In particular, we study the problem of determining the top-$k$ conditions that maximize the
reliability between two nodes. We deeply characterize our problem and show that, even employing polynomial-time reliability-estimation methods,
it is \NP-hard, does not admit any \PTAS, and the underlying objective function is non-submodular. We then devise a practical method that targets
both accuracy and efficiency. We also study natural generalizations of the problem with multiple source and target nodes.
An extensive empirical evaluation on several large, real-life graphs demonstrates effectiveness and scalability of our methods.
\end{abstract}

}

\maketitle
\sloppy
\IEEEdisplaynontitleabstractindextext
\IEEEpeerreviewmaketitle

\vspace{-2mm}
\section{Introduction}
\label{sec:intro}

\emph{Uncertain graphs}, i.e., graphs whose edges are assigned a probability of existence, have recently attracted a great
deal of attention, due to their rich expressiveness and given that uncertainty is inherent in the data in a wide range of applications.
Uncertainty may arise due to noisy measurements \cite{A09}, inference and prediction models \cite{AdarR07}, or explicit manipulation, e.g.,
for privacy purposes~\cite{Boldietal12}.
A fundamental problem in uncertain graphs is the so-called \emph{reliability},
which asks to measure the probability that two given (sets of) nodes are reachable \cite{AMG75}.
Reliability has been well-studied in the context of device networks, i.e.,
networks whose nodes are electronic devices and the (physical) links between such devices have a probability of failure \cite{AMG75}.
More recently, the attention has been shifted to other types of networks that can naturally be represented as uncertain graphs, such as social networks or
biological networks \cite{PBGK10,JLDW11}.

In the bulk of the literature,  reliability queries have been modeled without taking into account any external factor that
could influence the probability of existence of the links in the network.
In this paper, we overcome this limitation and introduce the notion of \emph{conditional reliability}, which takes into account that edge probabilities
may depend on a set of conditions, rather being fixed. This situation models real-world uncertain graphs. As an example, Figure~\ref{fig:example} shows a link $(u,v)$ of a \emph{social influence network}, i.e., a social graph where the associated probability represents the likelihood that a piece of information (e.g., a tweet) originated by $u$ will be ``adopted'' (re-tweeted) by her follower $v$.
The re-tweeting probability clearly depends on the content of the tweet.
In the example, $v$ is much more interested in sports than politics. Hence, if the tweet contains the hashtag \texttt{\#NFL},
then it will likely be re-tweeted by $v$,  while if it is about elections (i.e., it contains the hashtag \texttt{\#GetToThePolls}), it will be re-tweeted only with a small probability.
In this example, hashtags correspond to external factors that influence probabilities.
We hereinafter refer to such external factors as \emph{conditions} or \emph{catalysts}, and use all these terms interchangeably throughout the paper.

Given an uncertain graph with external-factor-dependent edge probabilities, in this work we study the following problem: Given a source node, a target node, and a small integer $k$, identify a set
of $k$ catalysts that maximizes the reliability between $s$ and $t$. This problem arises in many real-world scenarios, such as the ones described next.
\begin{figure}[t]
\vspace{-5mm}
\centering
\includegraphics[width=.7\columnwidth]{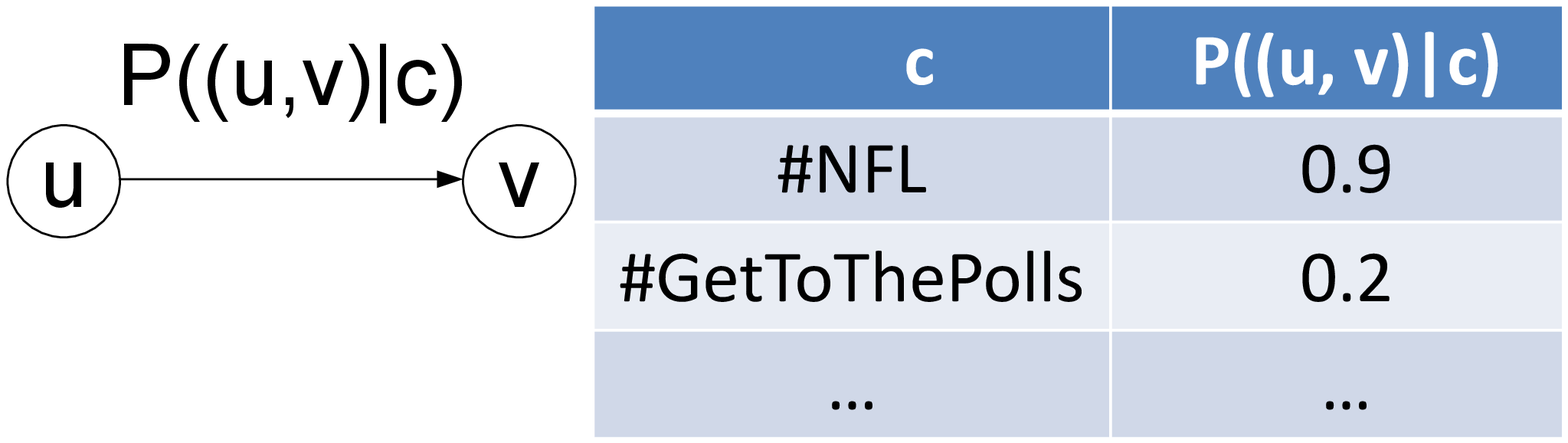}
\vspace{-2mm}
\caption{\scriptsize A link $(u,v)$ of a social influence network, where the associated probability represents the likelihood that a tweet by $u$ will be re-tweeted by her follower $v$. This probability depends on the content of the tweet. In this example if the tweet contains the hashtag \texttt{\#NFL}, then it will likely be re-tweeted, while if it is about elections (i.e., it contains the hashtag \texttt{\#GetToThePolls}), it will be re-tweeted only with a small probability.}
\label{fig:example}
\vspace{-5mm}
\end{figure}

\vspace{-1.5mm}
\spara{Pathway formation in biological networks.} To understand metabolic chain reactions in
cellular systems, biologists utilize metabolic networks \cite{JHWRX10}, where nodes represent compounds, and an edge between two
compounds indicates that a compound can be transformed into another one through a chemical reaction.
Reactions are controlled by various enzymes, and each enzyme defines a probability that the underlying reaction will actually take place.
Thus, reactions (edges) are assigned various probabilities of existence, which depend on the specific enzyme (external factor).
A fundamental question posed by biologists is to identify a set of enzymes which guarantee with high probability that a sequence of chemical reactions will take place to convert an input compound $s$ into a target compound $t$.
Since enzymes are expensive (they need to go through a long multi-step process before being commercialized \cite{CB90}), the output enzyme set should be limited in size.
Often known as {\em cost-effective experiment design} \cite{MYWY17,MGC04}, this corresponds to solving an instance of our problem: Given a source compound $s$ and a target compound $t$, what is the set of top-$k$ enzymes which maximizes the probability that $s$ will be converted into $t$ via a series of chemical reactions?

\vspace{-2mm}
\spara{Information cascades.} Studying information cascades in influential networks is receiving more and more attention, mainly due to its large applicability in {\em viral marketing}
strategies. Social influence can be modeled as in Figure~\ref{fig:example}, i.e.,  by means of a probability that once $u$ has been ``activated'' by a campaign, she will influence her friend
$v$ to perform the same action. This probability typically depends on topics and contents of the campaign \cite{BBM12,CFLFTT15}.
%As an example, it is highly desirable for any election candidate to evaluate the effect of her political standpoints, and judge which ones can influence people the most.
Within this view, let us consider the following example, which is motivated from \cite{LFZT17}. During the 2016 US Presidential election, Hillary Clinton's campaign promises were
infrastructure rebuild, free trade, open borders, unlimited immigration, equal pay, background checks to gun sales, increasing minimum wage, etc.
To get more votes, Hillary's publicity manager could have prioritized the most influential among all these standpoints in subsequent speeches from her, her vice presidential candidate
(Tim Kaine), and her political supporters (e.g., Barack and Michelle Obama), while also planning how to influence more voters from
the ``blue wall'' states (Michigan, Pennsylvania, and Wisconsin) \cite{S16}.
As speeches should be kept limited due to time constraints and risk of becoming ineffective in case of information overload, it is desirable to find a limited set of standpoints that maximize the influence from a set of early adopters (e.g., popular people who are close to
Hillary Clinton) to a set of target voters (e.g., citizens of the ``blue wall'' states)  \cite{aralwalker}.
% In similar contexts, it is crucial for a campaigner or a marketing company to identify the top-$k$ features that
%the campaign should possess, in order to maximize the information cascade from a set of early adopters (or influential users) to a group of target customers \cite{aralwalker}.
This corresponds to identifying the top-$k$ conditions that maximize the reliability between two (sets of) nodes in the social graph, i.e., the problem we study in this work.

\vspace{-1.5mm}
\spara{Challenges and contributions.}
The problem that we study in this work is a non-trivial one.
Computing standard reliability over uncertain graphs is a \sharpP-complete problem \cite{B86}.
We show that, even assuming polynomial-time sampling methods to estimate conditional reliability (such as  {\sf RHT}-sampling \cite{JLDW11}, recursive stratified sampling \cite{LYMJ14}),
our problem of computing a set of $k$ catalysts that maximizes conditional reliability between two nodes remains \NP-hard.
Moreover, our problem turns out to be not easy to approximate, as ($i$) it does not admit any \PTAS, and ($ii$) the underlying objective function is shown to be non-submodular.
Therefore, standard algorithms, such as iterative hill-climbing that greedily maximizes the marginal gain at every iteration,
do not provide any approximation guarantees and are expected to have limited performance.
Within this view, we devise a novel algorithm that first extracts highly-reliable paths between source and target nodes, and then iteratively selects these paths so as to achieve maximum improvement in reliability while still satisfying the constraint on the number of conditions.

After studying the single-source-single-target query, we focus on generalizations where multiple source and target nodes can be provided as input, thus opening the stage to a wider family of queries and applications.
We study two variants of this more general problem: (\emph{i}) maximizing an aggregate function over pairwise reliability between nodes in source and target sets, and (\emph{ii}) maximizing the probability that source and target nodes remain all connected.

The main contributions of this paper are as follows:
\squishlist
\item We focus on the notion of conditional reliability in uncertain graphs, which arises when the input graph has conditional edge-existence probabilities.
In particular, we formulate and study the problem of finding a limited set of conditions that maximizes reliability between a source and a target node (Section~\ref{sec:preliminaries}).

\item We deeply characterize our problem from a theoretical point of view, showing that it is \NP-hard and hard to approximate even when polynomial-time reliability estimation is employed (Section~\ref{sec:preliminaries}).

\item
We design an algorithm that provides effective (approximated) solutions to our problem, while also looking at efficiency.
The proposed method properly selects a number of highly-reliable paths so as to maximize reliability while satisfying the budget on the number of conditions (Section~\ref{sec:ourAlg}).

\item We generalize our problem and algorithms to the case of multiple source and target nodes (Section~\ref{sec:algorithm_multi}).

\item We empirically demonstrate effectiveness and efficiency of our methods on real-life graphs, while also detailing applications in information cascade (Section \ref{sec:experiments}).
\squishend

\vspace{-5mm}
\section{Single-source single-target: \\ problem statement}
\label{sec:preliminaries}

An uncertain graph $\mathcal{G}$ is a quadruple $(V, E, C, P)$, where $V$ is a set of $n$
nodes, $E \subseteq V \times V$ is a set of $m$ directed edges,  and $C$ is a set of external
conditions that influence the edge-existence probabilities.
We hereinafter refer to such external conditions as {\em catalysts}.
$P: E \times C \rightarrow (0,1]$ is a function that assigns a conditional probability to each edge $e \in E$ given a specific catalyst $c \in C$, i.e., $P(e|c)$ % = \Pr[e|c]$
denotes the probability that the edge $e$ exists given the catalyst $c$.

The bulk of the literature on uncertain graphs assumes that edge probabilities are independent of one another \cite{JLDW11}.
In this work, we make the same assumption.
Additionally, we assume that the existence of an edge is determined by an independent
process (coin flipping), one per catalyst $c$, and the ultimate existence of an edge is decided based on the success of at least one of such processes.
This assumption naturally holds in various settings.
For instance, in a metabolic network, with an initial
compound and an enzyme, the probability that a target compound would be produced depends only on that specific
reaction, and it is independent of other chemical reactions defined in the network.
As a result, the global
existence probability of an edge $e$, given a set of catalysts $C_1 \subseteq C$, can be derived as
$P(e|C_1) = 1- \prod_{c \in C_1}(1-P(e|c))$.

Given a set $C_1$ of catalysts, the uncertain graph $\mathcal{G}$ yields $2^m$ deterministic graphs $G \sqsubseteq \mathcal{G}|C_1$, where each $G$ is a pair $(V,$ $E_G)$, with $E_G \subseteq E$,
and its probability of being observed is given below.
\begin{align}
\displaystyle P(G|C_1) = \prod_{e\in E_G}P(e|C_1) \prod_{e\in E\setminus E_G}(1 - P(e|C_1))
\end{align}

For a source node $s \in V$, and a target node $t \in V$, we define \emph{conditional reliability} $R\left((s, t)|C_1\right)$ as the probability that $t$ is reachable from $s$ in $\mathcal{G}$, given a set $C_1$
of catalysts. Formally, for a possible graph $G \sqsubseteq \mathcal{G}|C_1$, let $I_G(s, t)$ be an indicator function taking value $1$ if there exists a path from $s$ to $t$ in $G$, and $0$ otherwise.
$R\left((s, t)|C_1\right)$ is computed as follows.
\vspace{-2mm}
\begin{align}
\displaystyle R\left((s, t)|C_1\right) = \sum_{G \sqsubseteq  \mathcal{G}| C_1} \left[I_G(s, t) \times P(G|C_1)\right]
\end{align}
The problem that we tackle in this work is introduced next.
\begin{figure}[t]
\vspace{-4mm}
\centering
\includegraphics[width=.2\columnwidth]{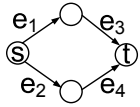}
\vspace{-2mm}
\caption{\scriptsize Example of non-submodularity. $P(e_1|c_1)=0.5$,
$P(e_2|c_2)=0.6$, $P(e_3|c_3)=0.5$, $P(e_4|c_1)=0.5$. $P(e|c)=0$ for all
other edge-catalyst combinations that are not specified.}\label{fig:nonsubmodular}
\vspace{-5mm}
\end{figure}

\begin{problem} [\topcat]
Given an uncertain graph $\mathcal{G}=(V,E,C,P)$, a source node $s \in V$, a target node $t \in V$, and a positive integer $k$,
find a set $C^* \subseteq C$ of catalysts, having size $k$, that maximizes the conditional reliability $R\left((s,t)|C^*\right)$ from $s$ to $t$:
\vspace{-4mm}
\begin{align}
& C^* = \argmax_{C_1 \subseteq C} R\left((s,t)|C_1\right) \nonumber & \\
& \text{subject to}\quad |C_1|=k. &
\vspace{-2mm}
\end{align}
\vspace{-3mm}
\label{prob:topk_rel_col_set}
\end{problem}
\vspace{-3mm}
Intuitively, the top-$k$ set $C^*$ yields multiple high-probability paths from the source node $s$ to the target node $t$.
Any specific path can have edges formed due to different catalysts.

\footnotetext[1]{\scriptsize Given an uncertain graph, a source node $s$, and a target node $t$, compute the probability that $t$ is reachable from $s$.}

\spara{Theoretical characterization.}
Problem \ref{prob:topk_rel_col_set} intrinsically relies on
the classical reliability problem \footnotemark[1], which is \sharpP-complete \cite{B86}.
As a result, Problem \ref{prob:topk_rel_col_set} is hard as well.

However, like standard reliability, conditional reliability can be estimated in polynomial time via Monte Carlo sampling, or other sampling methods \cite{JLDW11}.
Thus, the key question is whether Problem \ref{prob:topk_rel_col_set} remains hard even if polynomial-time conditional-reliability estimation is employed. As formalized next, the answer to this question is positive.
\begin{theorem}
\label{th:np_hard}
Problem \ref{prob:topk_rel_col_set} is \NP-hard even assuming polynomial-time computation for conditional reliability.
\end{theorem}
\begin{proof}
We prove \NP-hardness by a reduction from the \maxk problem.
In \maxk, we are given a universe $U$, and a set of $h$ subsets of $U$, i.e., $\mathcal{S} = \{S_1, S_2,\ldots, S_h\}$,
where $S_i \subseteq U$, for all $i \in [1\ldots h]$. The goal is to find a subset $\mathcal{S}^*$ of $\mathcal{S}$, of size $|\mathcal{S}^*|=k$,
such that the number of elements covered by $\mathcal{S}^*$ is maximized, i.e., so as to maximize $|\cup_{S \in \mathcal{S}^*} S|$.
Given an instance of \maxk, we construct in polynomial time an instance of \topcat problem as follows.

We create an uncertain graph ${\mathcal G}$ with a source node $s$ and a target node $t$. We add to $\mathcal{G}$ a set of nodes $u_1,u_2,\ldots,$ $u_Z$, one for each element
in $U$ ($Z = |U|$).
We connect each of these nodes $u_i$ to the target node $t$ with a (directed) edge $(u_i,t)$, and assume that each of such edges $(u_i,t)$ can occur only in the presence of
a single catalyst $c$ with a certain probability $p < 1$, i.e., $\forall i \in [1..Z]:P((u_i,t)|c) = p$ and $P((u_i,t)|c') = 0, \forall c' \neq c$.
Similarly, we put in $\mathcal{G}$ another set of nodes $x_1,x_2,\ldots,$ $x_Z$ (again one for each element in $U$),
and connect each of these nodes $x_i$ to the source node $s$ with an edge $(s,x_i)$. Each of such edges $(s,x_i)$ can also be present only in the presence of catalyst $c$, with probability $P((s,x_i)|c)=p$.
Finally, if some element $u_i\in U$ is covered by at least one of the subsets in $\mathcal{S}$, we add a directed edge $(x_i,u_i)$ in $\mathcal{G}$.
For each set $S_j \in \mathcal{S}$ that covers item $u_i$, we consider a corresponding catalyst $c_j$  and set the probability $P((x_i,u_i)|c_j)=1$.

Now, we ask for a solution of \topcat on the uncertain graph $\mathcal{G}$ constructed by using $k+1$ catalysts.
Every solution to our problem necessarily takes catalyst $c$, as otherwise there would be no way to connect $s$ to $t$.
Moreover, given that the paths connecting $s$ to $t$ are all disjoint, and each of them exists with probability $< 1$ (as $p <1$), the reliability from $s$ to $t$ is maximized by selecting $k$ additional catalysts that make the maximum number of paths exist, or, equivalently, selecting $k$ other catalysts that make each of the edges $(x_i, u_i)$ exist with probability $1$.
In order for each edge $(x_i, u_i)$ to exist with probability $1$, it suffices to have selected only one of the catalysts that are assigned to $(x_i, u_i)$.
Thus, selecting $k$ catalysts that maximize the number of edges $(x_i, u_i)$ existing with probability $1$ corresponds to selecting $k$ subsets $S_j$ that maximize the number of elements covered.
Hence, the theorem.
\end{proof}
\vspace{-2mm}
Apart from being \NP-hard, Problem \ref{prob:topk_rel_col_set} is also not easy to approximate, as it does not admit any \emph{Polynomial Time Approximation Scheme} (\PTAS).
\begin{theorem}
\label{th:apx_hard}
Problem \ref{prob:topk_rel_col_set} does not admit any \PTAS, unless \Poly\ = \NP.
\end{theorem}
\vspace{-5mm}
\begin{proof}
See Appendix.
\end{proof}

\vspace{-2mm}
As a further evidence of the difficulty of our problem, it turns out that neither submodularity nor supermodularity holds for the objective function therein.
Thus, standard greedy hill-climbing algorithms do not directly come with approximation guarantees.
Non-supermodularity easily follows from \NP-hardness (as maximizing supermodular set functions under a cardinality constraint is solvable in polynomial time), while non-submodularity is shown next with a  counter-example.
\begin{fact}
The objective function of Problem \ref{prob:topk_rel_col_set} is not submodular.
%\qed
\end{fact}
\vspace{-2mm}
A set function $f$ is submodular if $f(A\cup\{x\})-
f(A) \ge f(B \cup \{x\}) - f(B)$, for all sets $A \subseteq B$ and all elements $x \notin B$.
Look at the example in Figure~\ref{fig:nonsubmodular}. Let $C_1=\{c_2\}$,
$C_2=\{c_1,c_2\}$. We find that $R\left((s,t)|C_1\right)=0$,
$R\left((s,t)|C_1\cup\{c_3\}\right)=0$, $R\left((s,t)|C_2\right)=0.3$, and $R\left((s,t)|C_2\cup\{c_3\}\right)=0.475$.
Clearly, submodularity does not hold in this example.

\vspace{-5mm}
\section{Single-source single-target: baselines}
\label{sec:algorithm}

In this section, we present two simple baseline approaches and discuss their limitations (Sections \ref{sec:ind} and \ref{sec:hill_climb}).
Then, in Section \ref{sec:ourAlg}, we propose a more sophisticated algorithm that aims at overcoming the weaknesses of such baselines.

\vspace{-3mm}
\subsection{Individual top-$k$ baseline}
\label{sec:ind}

The most immediate approach to our \topcat problem consists of estimating the reliability $R\left((s, t)|\{c\}\right)$ between the source $s$ and the target $t$ attained by each catalyst $c \in C$ individually, and then outputting the top-$k$ catalysts that achieve the highest individual reliability.

\spara{Time complexity.} For each catalyst, we can estimate reliability via Monte Carlo ({\sf MC}) sampling\footnotemark[2]: sample a set of $K$ deterministic graphs from the input uncertain graph, and estimate reliability by summing the (normalized) probabilities of the graphs where the target is reachable from the source. The time complexity of {\sf MC} sampling for a single catalyst is $\bigO(K(n+m))$, where $n$ and $m$ denote the number of nodes and edges in the input uncertain graph, respectively. Hence, the overall time complexity of the \bone baseline is $\bigO(|C|K(n+m)+|C|\log k)$, where the last term is due to top-$k$ search.

\spara{Shortcomings.} The \bone algorithm suffers from both accuracy and efficiency issues.
\begin{itemize}
\setlength\itemsep{0.1em}
\item Accuracy: This baseline is unable to capture the contribution of paths containing different catalysts. For example, in Figure~\ref{fig:nonsubmodular}, the individual reliability attained by each catalyst is $0$.
Thus, if we are to select the top-$2$ catalysts, there will be no way to discriminate among catalysts, which will be picked at random. Instead, in reality, the top-$2$ set is $\{c_1, c_2\}$.
\item Efficiency: To achieve good accuracy, {\sf MC} sampling typically requires around thousands of samples \cite{JLDW11}. Performing such a sampling for each of the $|C|$ catalysts can be quite expensive on large graphs ($|C|$ may be up to the order of thousands as well, see Section \ref{sec:experiments}).
\end{itemize}

\vspace{-5mm}
\subsection{Greedy baseline}
\label{sec:hill_climb}

A more advanced baseline consists of greedily selecting the catalyst that brings the maximum marginal gain to the total reliability, until $k$ catalysts have been selected.
More precisely, assuming that a set $C_1$ of catalysts has been already computed, in the next iteration this \btwo baseline selects a catalyst $c^*$ such that:
\begin{equation*}
c^* = \argmax_{c \in C\setminus C_1} \ [R\left((s,t)|C_1\cup\{c\}\right)-R\left((s,t)|C_1\right)]
\end{equation*}
Note that, since the \topcat problem is neither submodular nor supermodular, this greedy approach does not achieve any approximation guarantees.

\spara{Time complexity.} The time complexity of each iteration of the greedy baseline is $\bigO(|C|K(n+m))$, as we need to estimate the reliability achieved by the addition of each catalyst in order to choose the one maximizing the marginal gain.
For a total of $k$ iterations (top-$k$ catalysts are to be reported), the overall complexity is $\bigO(|C|kK(n+m))$.

\spara{Shortcomings.} While being more sophisticated than \bone, the \btwo baseline still suffers from both accuracy and efficiency issues.
\begin{itemize}
\setlength\itemsep{0.1em}
\item Accuracy: Although \btwo partially solves the accuracy issue related to the presence of paths with multiple catalysts, such an issue is still present at least in the initial phases of this second baseline.
For example, in Figure~\ref{fig:nonsubmodular} the individual reliability attained by each catalyst is $0$.
Therefore, in the first iteration the \btwo algorithm has no information to properly select a catalyst, thus ending up with a completely random choice. If $c_3$ is selected as a first catalyst, then the second catalyst selected would be $c_1$. Thus, \btwo would output $\{c_1,c_3\}$, while the top-$2$ set is $\{c_1,c_2\}$. We refer to this issue as {\em ``cold-start''} problem.
\item Efficiency: {\sf MC} sampling is performed $|C|k$ times. This is more inefficient than the \bone.
\end{itemize}
\begin{example}
We demonstrate the cold-start problem associated with the \btwo baseline with a running example in Figure~\ref{fig:inclusion1}.
Assume top-$k$=3. The individual reliability from $s$ to $t$, attained by each of the four catalysts is $0$.
Therefore, in the first iteration, the \btwo algorithm selects a catalyst uniformly at random, say $c_4$.  Then, the second catalyst selected would be $c_1$;
since $c_1$, in the presence of $c_4$, provides the maximum marginal gain compared to any other catalyst. Similarly, in the third round, \btwo
will select $c_2$ due to its higher marginal gain. Therefore, total reliability achieved by \btwo is: $R\big((s,t)|\{c_4,c_1,c_2\}\big)=1-(1-0.5\times0.5)(1-0.8\times0.7\times0.7)=0.544$.
However, the top-$3$ set is $\{c_1,c_2,c_3\}$, yielding reliability $R\big((s,t)|\{c_1,c_2,c_3\}\big)=0.8[1-(1-0.8)(1-0.7\times0.7)]=0.7184$. This shows the sub-optimality of the greedy baseline.
\end{example}
\begin{figure}[t]
\vspace{-4mm}
\centering
\includegraphics[scale=0.23]{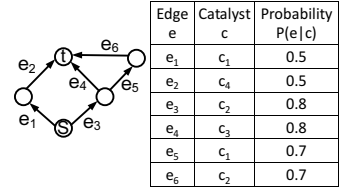}
\vspace{-2mm}
\caption{\scriptsize Difficulties with the \btwo baseline. $P(e|c)=0$ for all
edge-catalyst combinations that are not present in the table}\label{fig:inclusion1}
\vspace{-5mm}
\end{figure}

\vspace{-5mm}
\section{Single-source single-target: \\ proposed method}
\label{sec:ourAlg}

Here we describe the method we ultimately propose to provide effective and efficient solutions to the \topcat problem.

The main intuition behind our method directly follows from the shortcomings of the two baselines discussed above.
Particularly, both baselines highlight how considering catalysts one at a time is less effective.
This can easily be explained as a single catalyst can bring information that is related only to single edges.
Instead, what really matters in computing the reliability between two nodes is the set of paths connecting the source and the target.
This observation finds confirmation in the literature \cite{CWW10}.

\footnotetext[2]{\scriptsize In this paper, we employ {\sf MC} sampling as an oracle
to estimate reliability in uncertain graphs. While more advanced sampling techniques exist, e.g., {\sf RHT} \cite{JLDW11}, recursive stratified sampling \cite{LYMJ14}, our
contributions are orthogonal to them. We omit discussing advanced sampling methods for brevity.}

Motivated by this, we design the proposed method as composed of two main steps.
First, we select the top-$r$ paths exhibiting highest reliability from the source to the target.
Second, we iteratively include these paths in the solution so as to maximize the marginal gain in reliability, while still keeping the constraint on total number of catalysts satisfied.
Apart from the main advantage due to considering paths instead of individual catalysts, designing our algorithm as composed of two separate steps allows us to achieve high efficiency.
Indeed, the first step can be efficiently solved by fast algorithms for finding the top-$r$ shortest paths, while the second step requires \textsf{MC} sampling to be performed in a significantly reduced version of the original graph.

The outline of the proposed method, which we call \ours, is reported in Algorithm \ref{alg:ours}.
In the following we provide the details of each of the two steps.

\vspace{-4mm}
\subsection{Step 1: Most-reliable paths selection}

The first step of the proposed method consists of finding the top-$r$ most reliable paths from the source to the target.
Given an uncertain graph $\mathcal{G}=(V,E,C,P)$, a source node $s \in V$, and a target node $t \in V$, we first convert $\mathcal{G}$ into an uncertain, multigraph $\mathcal{G'}$ (Algorithm~\ref{alg:reliable_path}). For each edge $e = (u,v) \in E$, let $C(e) \subseteq C$ denote the set of all \emph{single} catalysts such that $\forall c\in C(e): P(e|c) > 0$.
Assume $C(e) = \{c_1,c_2,\ldots,c_i\}$.
Then, we add $i$ edges $\{e_1, e_2,$ $ \ldots, e_i\}$ between $u$ and $v$ in the multigraph $\mathcal{G'}$.
To each newly constructed edge $e_j$, $j \in [1..i]$, we assign a single catalyst $C(e_j)=c_j$ and set $P(e_j|c_j)= P(e|c_j)$.
It can be easily noted that $\mathcal{G}$ and $\mathcal{G'}$ are equivalent in terms of our problem.
The construction of $\mathcal{G'}$ only serves the purpose of selecting the top-$r$ most reliable paths from $s$ to $t$ in such a way that, for each intermediate pair of nodes $x,y$ along a path, a single edge (and, thus, a single catalyst) among the many ones possibly created by the $\mathcal{G} \rightarrow \mathcal{G'}$ transformation, is picked up.
The reliability of a path is defined as the product of the edge-probabilities along that path.
\begin{algorithm}[tb!]
\caption{\ours}
\small
\begin{algorithmic}[1]\label{alg:ours}
\REQUIRE Uncertain graph $\mathcal{G} = (V,E,C,P)$, source node $s \in V$, target node $t \in V$, positive integers $k, r$
\ENSURE Subset of catalysts $C^* \subseteq C$
\STATE $\mathcal{P} \gets $ Algorithm \ref{alg:reliable_path} on input $(\mathcal{G}, s, t, r)$
\STATE $\mathcal{P}_1 \gets $ Algorithm \ref{alg:greedy_path} on input $(\mathcal{G}, s, t, k)$, $\mathcal{P}$
\STATE $C^* \gets $catalysts present on $\mathcal{P}_1$
\end{algorithmic}
\vspace{-1mm}
\end{algorithm}
\begin{algorithm}[tb!]
\caption{{\sf Top-$r$ Most Reliable Paths Selection}}
\small
\begin{algorithmic}[1]\label{alg:reliable_path}
\REQUIRE Uncertain graph ${\mathcal G}=(V,E,C,P)$, source node $s \in V$, target node $t \in V$, positive integer $r$
\ENSURE $\mathcal{P}$: top-$r$ most reliable paths from $s$ to $t$
\FORALL{$e \in E$}
\STATE let $C(e)=\{c_1,c_2,\ldots,c_i\}$ be the set of all catalysts s.t. $P(e|c_j)>0$,  $\forall j \in [1..i]$
\STATE replace $e$ by $i$ edges $\{e_1, e_2, \ldots, e_i\}$
\STATE assign probability $P(e_j|c_j)=P(e|c_j)$
\STATE assign edge-weight $W(e_j) = -\log P(e_j|c_j)$
\ENDFOR
\STATE $\mathcal{P} \leftarrow$ top-$r$ shortest paths from $s$ to $t$ in the constructed multigraph
\end{algorithmic}
\vspace{-1mm}
\end{algorithm}
\begin{algorithm}[tb!]
\caption{\steptwo}
\small
\begin{algorithmic}[1]\label{alg:greedy_path}
\REQUIRE Top-$r$ most-reliable path set $\mathcal{P}$ from source $s$ to target $t$, positive integer $k$
\ENSURE A subset of paths ${\mathcal P}_1 \subseteq {\mathcal P}$
\STATE $\mathcal{P}_1 \gets \emptyset$
\WHILE{$|\mathcal{P}| > 0$ \ and \ total \#catalysts in ${\mathcal P}_1$ is $\leq k$}
\STATE $P^* \gets \arg\max_{P \in \mathcal{P}\setminus \mathcal{P}_1} Rel_{\mathcal{P}_1\cup\{P\}}(s,t)$ \\ s.t. \#catalysts in ${\mathcal P}_1 \cup \{P^*\}$ is $\leq k$
\STATE $\mathcal{P}_1 \gets \mathcal{P}_1 \cup \{P^*\}$
\STATE $\mathcal{P} \gets \mathcal{P} \setminus \{P^*\}$
\ENDWHILE
\end{algorithmic}
\vspace{-1mm}
\end{algorithm}

To ultimately compute the top-$r$ most reliable paths, we further convert the uncertain multigraph $\mathcal{G'}$ into an edge-weighted multigraph $\mathcal{G''}$ by assigning a weight $-\log(p_e)$ to each edge $e$ with probability $p_e$ of $\mathcal{G}'$.
This way, the top-$r$ most reliable paths in $\mathcal{G'}$ will correspond to the top-$r$ shortest paths in $\mathcal{G''}$.
To compute the top-$r$ shortest paths in $\mathcal{G''}$, we apply the well-established Eppstein's algorithm  \cite{E98},
which has time complexity $\bigO(|C|m + nlogn + r)$.

\spara{Space complexity.} We note that both $\mathcal{G'}$ and $\mathcal{G''}$ can have size at most $|C|$ times more than
the size of the original graph $\mathcal{G}=(V,E,C,P)$. This is because in Algorithm~\ref{alg:reliable_path}, each edge $e$ of $\mathcal{G}$
is replaced by $C(e)$ edges in $\mathcal{G'}$ and $\mathcal{G''}$ (lines 2-3),
where $C(e)$ denotes the set of all catalysts such that $\forall c \in C(e), P(e|c) > 0$. Clearly, $C(e)\le |C|$. Therefore, both $\mathcal{G'}$ and $\mathcal{G''}$
can have at most $|E||C|=m|C|$ edges, while still having the same number of nodes as in the original graph. In summary, the size of $\mathcal{G'}$ and $\mathcal{G''}$
is {\em linear} in that of the original graph and in the number of catalysts. Based on our experimental results, this adds a very small overhead to the overall
space requirement (see Section~\ref{sec:experiments}).

\spara{Choice of $r$.} The number $r$ of paths is an input parameter which constitutes a knob to tradeoff between efficiency and accuracy (a larger $r$ leads to higher accuracy and lower efficiency).
In general, its choice depends on the input graph and the number of top-$k$ catalysts.
An easy yet effective way to set it is to observe when the inclusion of the top-$(r+1)$-th reliable path does not tangibly increase the reliability given by the top-$r$ paths.
We provide experimental results on selecting $r$ in Section~\ref{sec:experiments}.

\vspace{-3mm}
\subsection{Step 2: Iterative path inclusion}

The second step of our \ours method aims at selecting a proper subset from the top-$r$ most-reliable path set so as to maximize reliability between source and target nodes, while also meeting the constraint on the number of output catalysts.
Denoting by $Rel_{\mathcal{P}}(s,t)$ the reliability between $s$ and $t$ in the subgraph induced by a path set $\mathcal{P}$, this step formally corresponds to the following problem:
\begin{problem} [\ipi]
Given set $\mathcal{P}$ of top-$r$ most reliable paths from $s$ to $t$ in multigraph $\mathcal{G'}$, find a path set $\mathcal{P}^* \subseteq \mathcal{P}$
such that:
\vspace{-3mm}
\begin{align}
\mathcal{P}^* = \argmax_{\mathcal{P}_1 \subseteq \mathcal{P}} Rel_{\mathcal{P}_1}(s,t) \nonumber \\
\text{subject to}\quad | \textstyle{\bigcup_{e\in \mathcal{P}_1} C(e)} | \le k
\vspace{-2mm}
\end{align}
\vspace{-4mm}
\label{prob:path}
\end{problem}

\vspace{-2mm}
The \ipi problem  can be shown to be  \NP-hard via  a reduction from \maxk.
The proof is analogous to the one in Theorem \ref{th:np_hard}, we thus omit it.
\begin{theorem}
Problem \ref{prob:path} is \NP-hard.
\label{th:np_path}
\end{theorem}

\vspace{-3mm}
\spara{Algorithm.}
We design an efficient greedy algorithm (Algorithm~\ref{alg:greedy_path}) for the \ipi problem.
At each iteration, we add a path $P^*$ to the already computed path set $\mathcal{P}_1$ which brings the maximum marginal gain in terms of reliability.
While selecting path $P^*$, we also ensure that the total
number of catalysts used in the paths $\mathcal{P}_1\cup\{P^*\}$ is no more than $k$. The algorithm terminates either when there is no path left in the top-$r$ most reliable path set $\mathcal{P}$, or no more paths can be added without violating the budget $k$ on the number of catalysts.
We report the catalysts present in $\mathcal{P}_1$ as our final solution.
If the total number of catalysts present in $\mathcal{P}_1$ is $k'<k$, additional $k-k'$ catalysts that are not in $\mathcal{P}_1$ can be selected with some proper criterion (e.g., frequency on the non-selected paths).

Next, we demonstrate our \steptwo algorithm with the previous running example.
\begin{example}
In Figure~\ref{fig:inclusion11}, which is same as Figure~\ref{fig:inclusion1},
we have 3 paths from $s$ to $t$, i.e., $P_1: e_1e_2$, $P_2:e_3e_4$, and $P_3:e_3e_5e_6$.
Assume that there is a budget constraint of 3 catalysts.
In the first iteration, we select the path $P_2$ since it has the highest reliability compared to the two other paths.
In the second iteration, $P_1$ and $P_2$ together have higher reliability than $P_2$ and $P_3$.
However, the former combination requires 4 catalysts, thus violating the constraint.
Hence, we select $P_2$ and $P_3$.
After that, the algorithm terminates as no more path can be included without violating the constraint on catalysts.
\end{example}

\vspace{-2mm}
\spara{Approximation guarantee.} The \steptwo algorithm achieves approximation guarantee under some assumptions.
If the top-$r$ most reliable paths are node-disjoint (except at source and target nodes), \steptwo exhibits an approximation ratio at least $\frac{1}{r}$.
\begin{figure}[tb!]
\subfigure [\small 3-paths from $s$ to $t$] {
\includegraphics[scale=0.145]{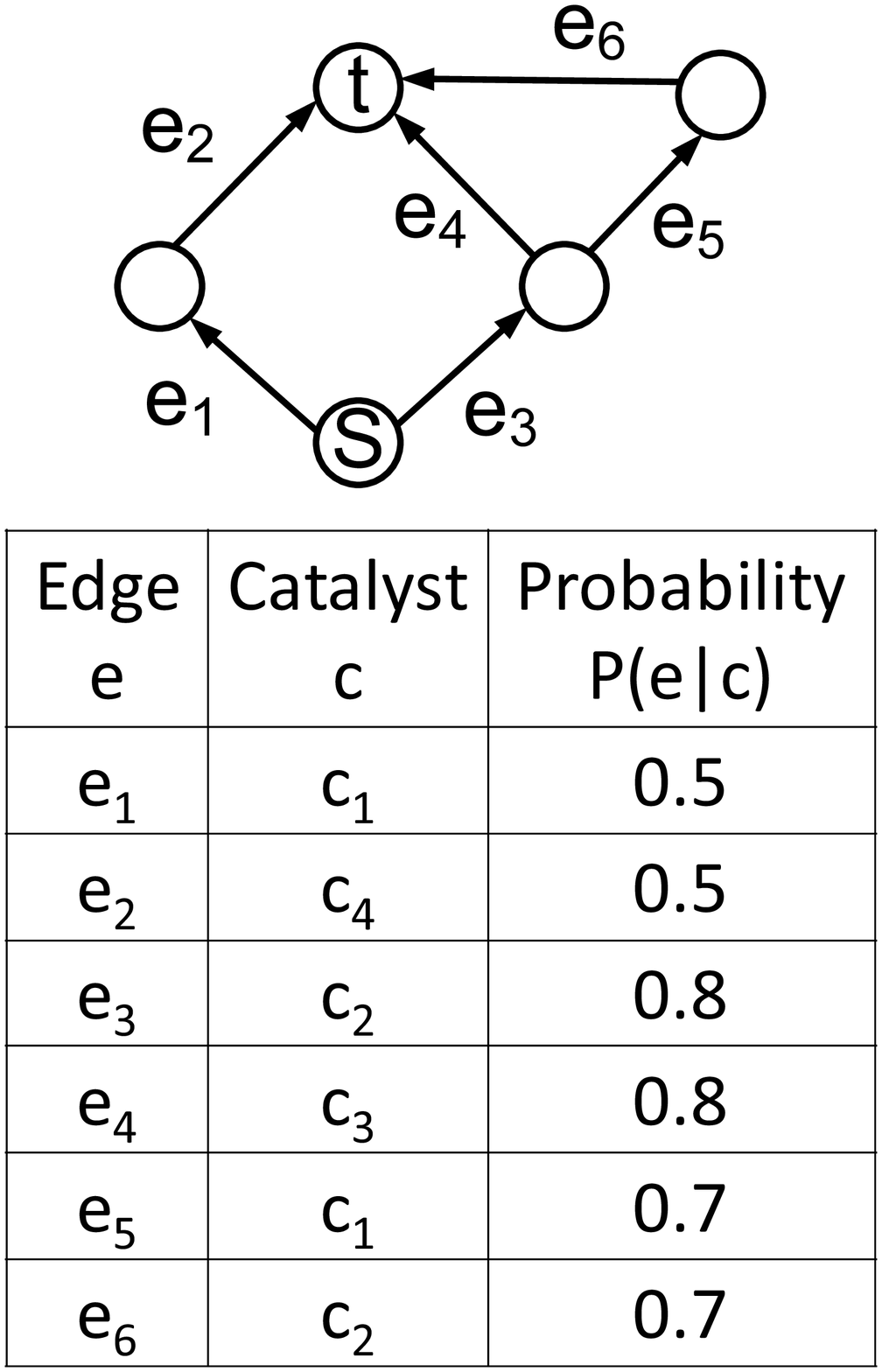}
\label{fig:inclusion11}
}
$\quad$
\subfigure [\small 1$^\text{st}$ Iteration] {
\includegraphics[scale=0.145]{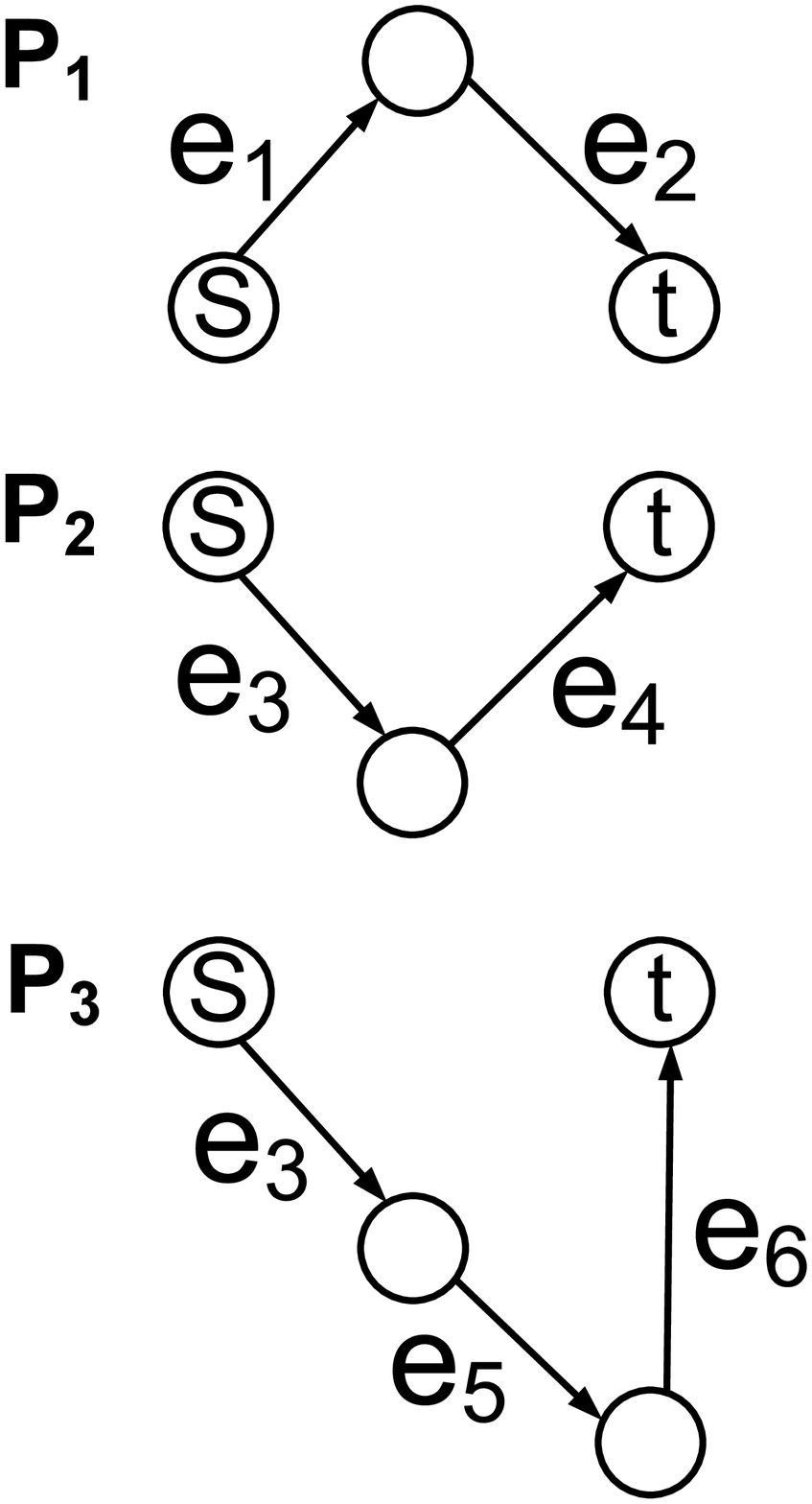}
\label{fig:inclusion2}
}
$\quad$
\subfigure [\small 2$^\text{nd}$ Iteration] {
\includegraphics[scale=0.145]{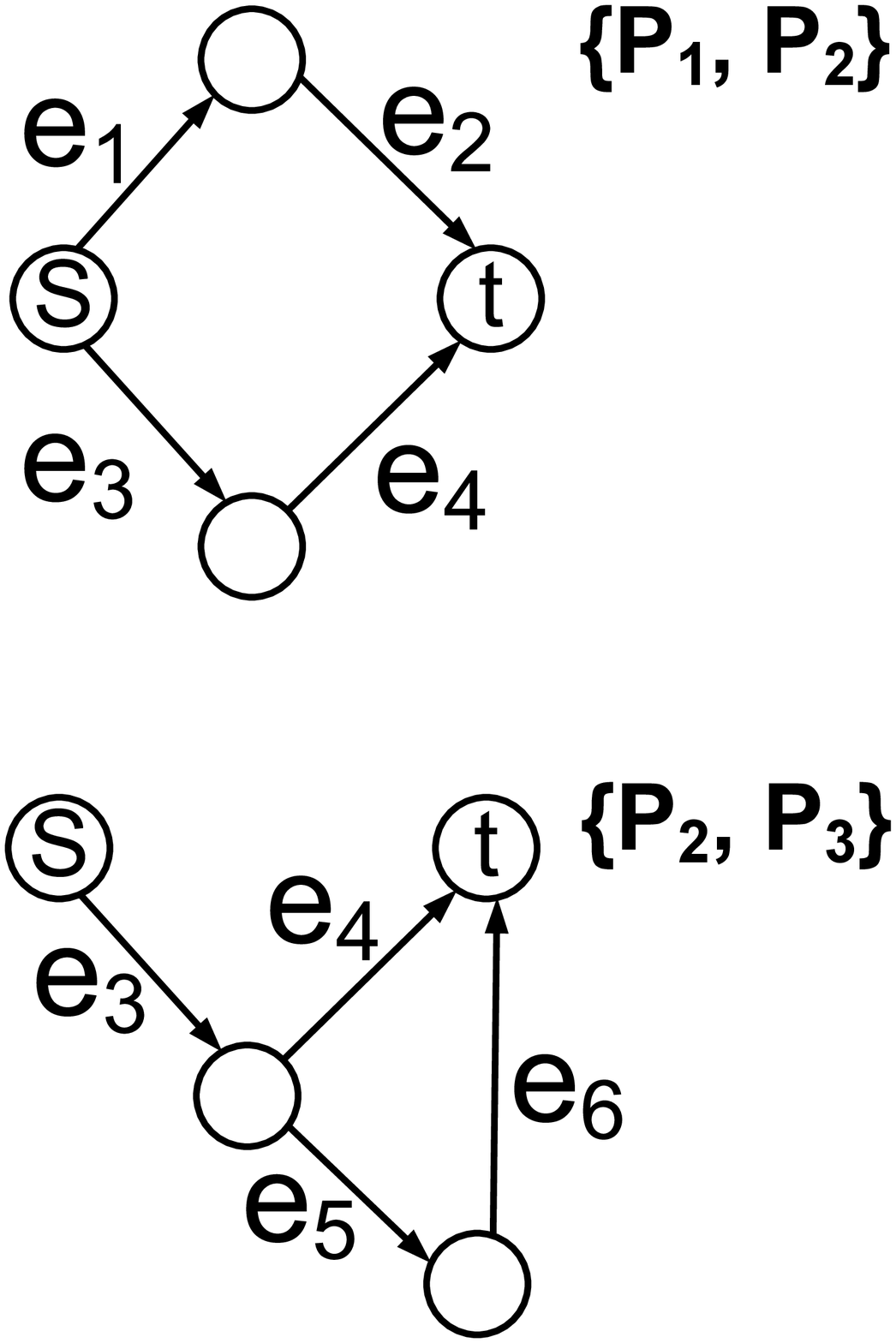}
\label{fig:inclusion3}
}
\vspace{-3mm}
\caption{\scriptsize A demonstration of the \steptwo algorithm.}
\label{fig:iter_path_inc}
\vspace{-5mm}
\end{figure}

\begin{theorem}
The \steptwo algorithm, under the assumption that the top-$r$ most reliable paths are node-disjoint,
achieves an approximation factor of:
\begin{align}
\displaystyle \frac{1}{k_{Rel}}\left(1-\left(\frac{K_C-k_{Rel}}{K_C}\right)^{k_C}\right),
\end{align}
\vspace{-4mm}
where

\begin{align}
& K_C = \max_{\mathcal{P}_1\subseteq\mathcal{P}}\left \{\left|\mathcal{P}_1\right|:\left|C\left(\mathcal{P}_1\right)\right|\le k \right \} & \\
& k_C = \min_{\mathcal{P}_1\subseteq\mathcal{P}}\left \{\left|\mathcal{P}_1\right|:\left|C\left(\mathcal{P}_1\right)\right|\le k\} \quad \text{and} \right. \nonumber & \\
& \left. \qquad \qquad \forall P \in \mathcal{P}\setminus\mathcal{P}_1, \quad |C(\mathcal{P}_1\cup\{P\})|> k \right \} &
\end{align}
\begin{align}
& k_{Rel} = 1-\min_{P \in \mathcal{P}} \frac{Rel_{\mathcal{P}}(s,t)-Rel_{\mathcal{P}\setminus\{P\}}(s,t)}{Rel_{\{P\}}(s,t)}.&
\end{align}
\label{th:guarantee}
\end{theorem}
\vspace{-6mm}
\begin{proof}
See Appendix.
\end{proof}
In the above approximation-guarantee result, $K_C$ and $k_C$, respectively, denote maximum and minimum size of the maximal feasible path set that can be derived from $\mathcal{P}$.
$k_{Rel}$ denotes the curvature of our optimization function, which can be shown to be submodular when paths in $\mathcal{P}$ are node-disjoint (see Appendix).
Hence, in this case $k_{Rel} \in (0,1)$.
Assuming that $\mathcal{P}$ contains at least one path having less than $k$ catalysts, then in the worst case the approximation ratio is $ \ge \frac{1}{K_C} \ge \frac{1}{r}$ (where $r$ is the total number of paths in the top-$r$ path set $\mathcal{P}$).
In other words, the approximation ratio is guaranteed to be at least $\frac{1}{r}$.

\spara{Time complexity.} Let us denote by $n'$ and $m'$ the number of nodes and edges, respectively, in the subgraph induced by the top-$r$ most-reliable
path set $\mathcal{P}$.
At each iteration, our iterative path selection algorithm performs
{\sf MC} sampling over the subgraph induced by the selected paths.
The number of iterations is at most $r$. Thus, if $K$ is the number of samples used in each {\sf MC} sampling, the iterative path selection algorithm takes
$\bigO(r^2K(n'+m'))$ time.
Including the time due to the first step of selecting the top-$r$ most reliable paths, we get that the overall time complexity of the proposed \ours algorithm is $\bigO(|C|m + n\log n + r^2K(n'+m'))$.
We point out that the subgraph induced by the top-$r$ most reliable paths is typically much smaller than
the input graph $\mathcal{G}$.
Thus, our \ours method is expected to be much more efficient than the two baselines introduced earlier.
Experiments in Section \ref{sec:experiments} confirm this claim.

\vspace{-4mm}
\section{Multiple sources and targets}
\label{sec:algorithm_multi}

Real-world queries often involve sets of source and/or target nodes,  instead of a single source-target pair.
As an example, the topic-aware information cascade problem \cite{aralwalker} asks for a {\em set} of early adopters who maximally
influence a given {\em set} of target customers.
Motivated by this, in the following we discuss problems and algorithms for the case where multiple nodes can be provided as input.
Such a generalization opens the stage to various formulations of the problem.
Here we focus on two variants: (1) maximizing an aggregate function over all possible source-target pairs (Section~\ref{sec:agg_set}),
and (2) maximizing connectivity among all query nodes (Section~\ref{sec:connectivity}). Note that our first problem formulation has a notion of ``clique'' connectivity, as it applies an aggregate function over all pairs of query nodes.

\vspace{-4mm}
\subsection{Maximizing aggregate functions}
\label{sec:agg_set}
%
%\spara{Problem statement.}

We formulate our problem as follows.

\begin{problem} [top-$k$ catalysts w/ aggregate]
Given an uncertain graph $\mathcal{G}=(V,E,C,P)$, a source set $S \subset V$, a target set $T \subset V$, and a positive integer $k$, find a set $C^*$ of catalysts, having size $k$ that maximizes an aggregate function $F$ over conditional reliability of all source-target pairs:
\vspace{-1mm}
\begin{align}
& \displaystyle C^* = \argmax_{C_1 \subseteq C} \underset{{\langle s, t \rangle \in S \times T}}{F}\big(R\left((s,t)|C_1\right)\big) \nonumber & \\
& \text{subject to}\quad |C_1|=k. &
\vspace{-2mm}
\end{align}
\vspace{-4mm}
\label{prob:agg_set}
\end{problem}
\vspace{-2mm}
Being a generalization of the \topcat problem, Problem \ref{prob:agg_set} can easily be recognized as \NP-hard.
In this work we consider three commonly-used aggregate functions: average, maximum, and minimum.
These aggregates give rise to three variants of Problem \ref{prob:agg_set} which we refer to \topcatavg, \topcatmax, and \topcatmin, respectively.

\vspace{-1mm}
\spara{Motivating examples for multiple sources and targets.}
\vspace{-1mm}
\begin{itemize}
\setlength\itemsep{0.1em}
\item {\bf Average.}
Find the top-$k$ catalysts such that the average reliability over all $\langle s, t \rangle$ pairs is maximized. This is equivalent
to maximization of the sum of reliability over all $\langle s, t \rangle$ pairs.  This problem occurs, e.g., in the topic-aware information cascade scenario
when the campaigner wants to maximize the spread of information to the entire target group.
\item {\bf Maximum.} Find the top-$k$ catalysts such that the reliability of the $\langle s, t \rangle$ pair with the highest reliability is maximized.
In the topic-aware information cascade problem, this is equivalent to the scenario that each early adopter is campaigning a different product of the same campaigner. The campaigner
wants at least one target user to be aware about one of her products (e.g., each target user might be a celebrity user in Twitter).
Therefore, the campaigner would be willing to maximize the spread of information from at least one early adopter to at least one target user.
\item {\bf Minimum.} Find the top-$k$ catalysts such that the reliability of the $\langle s, t \rangle$ pair having the lowest reliability is maximized. In the
topic-aware information cascade setting, this is equivalent to the problem that each early adopter is campaigning a different product of the same campaigner, and the campaigner wants to maximize the minimum spread of her campaign from any of the early adopters to any of her target users. This is motivated, in reality, because only a small percentage of the users who have heard about a campaign will buy the corresponding product.
\end{itemize}

\vspace{-1mm}
\spara{Overview of algorithms.}
In the following, we describe the algorithms that we develop for the aforementioned aggregate functions.
In principle, they follow our earlier two main steps: (1) finding the top-$r$ paths (Algorithm~\ref{alg:reliable_path}),
now between every pair of source and target
nodes, and then (2) iteratively include these paths so as to maximize the marginal gain in regards to the objective function,
while still keeping the constraint on total number of catalysts satisfied (Algorithm~\ref{alg:greedy_path}).
Intuitively, finding the top-$r$ paths between each pair of source and target
nodes is a natural extension to our \ours algorithm, this is because the
aggregate function in Problem~\ref{prob:agg_set}
is defined over all pairs of source-target nodes.
Nevertheless, the exact process somewhat varies
according to the aggregate function at hand, which we shall discuss next. Unless otherwise specified, we assume that
\mbox{$S \cap T = \emptyset$}, that is, source and target sets are non-overlapping.
We will discuss case by case how our algorithms can (easily) handle the case when \mbox{$S \cap T \ne \emptyset$}.

Extending the baselines presented in Sections \ref{sec:ind}--\ref{sec:hill_climb} to multiple query nodes
is instead straightforward (regardless of the aggregate function). We thus omit the details.

\vspace{-2mm}
\subsubsection{Algorithm for maximum reliability}
\label{sec:alg_max}

Our solution for the \topcatmax problem is the most straightforward, compared to both \topcatavg and \topcatmin
problems. First, for each $\langle s, t \rangle$ pair, we identify the top-$r$ most reliable paths. Then, separately for
each $\langle s, t \rangle$ pair, we also apply the \steptwo algorithm to find the top-$k$ catalysts for that pair.
Finally, we select the $\langle s, t \rangle$ pair which attains the maximum reliability.
We report the corresponding top-$k$ catalysts as the solution to the \topcatmax problem.

\spara{Time complexity.} The time required to find the reliable paths for all $\langle s, t \rangle$ pairs is $\bigO\left(|S||T|\left(m+n \log n+r\right)\right)$. Similarly, the time complexity of the iterative path inclusion phase is $\bigO\left(|S||T|r^2\left(n' + m'\right)K\right)$. There is an additional cost $\bigO(|S||T|)$ to find the $\langle s, t \rangle$ pair with the maximum reliability, which is, however, dominated by the time spent in path inclusion.

\spara{Overlapping source and target sets.} If a node $v$ is both in $S$ and in $T$, the above solution will return an arbitrary set $C_1$ of catalysts, since $R\left((v,v)|C_1\right) = 1$, and it will always be considered as the optimal solution. If this behavior is not intended, we eliminate all such nodes in $S \cap T$ before applying our algorithm.

\vspace{-2mm}
\subsubsection{Algorithm for average reliability}
\label{sec:alg_avg}

\vspace{-1mm}
As earlier, we first identify the top-$r$ most reliable paths for each $\langle s, t \rangle$ pair.
However, we are now interested in the {\em average} reliability considering all source and target
nodes, as opposed to that for individual source-target pairs. Thus, we consider {\em all} selected
$|S||T|r$ paths together, and apply the \steptwo algorithm to add paths so as to maximize
the marginal gain in terms of our objective function, while maintaining the budget
$k$ on total number of catalysts. Recall that here our objective function is
$\frac{1}{|S||T|}\sum_{\langle s, t \rangle \in S \times T}\big(R\left((s,t)|C_1\right)\big)$.
Finally, catalysts present in the selected paths are reported as a solution to the \topcatavg problem.

The above steps remain identical, regardless of whether $S$ and $T$
overlap or not.

\spara{Time complexity.} The time required to find the reliable paths for all $\langle s, t \rangle$ pairs is $\bigO\left(|S||T|\left(m+n \log n+r\right)\right)$,
as we apply Eppstein's algorithm for $|S||T|$ times. Next, the time complexity of the iterative path inclusion phase is
$\bigO\left(\left(|S||T|r\right)^2\left(n' + m'\right)K\right)$, where $n'$ and $m'$ are the number of nodes and edges of the subgraph induced by the reliable paths, and $K$ is the number of samples used in each {\sf MC} sampling. Note that the time required for
iterative path inclusion of \topcatavg is higher than that for the \topcatmax, since we consider
all $|S||T|r$ paths together in the former algorithm.

\vspace{-2mm}
\subsubsection{Algorithm for minimum reliability}
\label{sec:alg_min}

\vspace{-1mm}
We start again by finding the top-$r$ most reliable paths for each $\langle s, t \rangle$ pair.
However, applying the \steptwo algorithm, in this case, is more subtle.
Specifically, if there are many $\langle s, t \rangle$ pairs and a limited budget $k$ of catalysts, spending too many catalysts for a single $\langle s, t \rangle$ pair might prevent us from finding a path for another pair.
This way there will be pairs with conditional reliability very low, thus the solution to the \topcatmin problem, i.e., the pair exhibiting minimum conditional reliability, would be quite poor.
To mitigate this issue, we consider an additional step where we find a minimum set of catalysts, before applying the \steptwo algorithm. The subsequent steps remain instead identical, regardless of whether $S$ and $T$ overlap or not.

\spara{Finding minimum set of catalysts.} The objective of this step is to select a minimum set of catalysts which ensure that at least one path
for every $\langle s, t \rangle$ pair exists.
This step corresponds to the following problem.
\begin{problem}[\mincat]
Given a source set $S$, a target set $T$, a set of paths ${\mathcal P}$, an uncertain graph ${\mathcal G} = (V, E, C, P)$ induced by ${\mathcal P}$, find the
smallest set $C^*\subseteq C$ of catalysts, such that the conditional
reliability $R\left((s,t)|C^*\right)$ for each $\langle s, t \rangle$ pair is larger than zero:
\begin{align}
& \displaystyle C^* = \argmin_{C_1 \subseteq C} |C_1| \nonumber & \\
& \text{subject to} \quad  R\left((s,t)|C_1\right) > 0, \ \forall \langle s, t \rangle \in S \times T. &
\vspace{-2mm}
\end{align}
\vspace{-2mm}
\label{prob:col_set}
\end{problem}
\vspace{-2mm}

\vspace{-4mm}
\begin{theorem}
Problem \ref{prob:col_set} is \NP-hard.
\end{theorem}
\begin{proof}
\NP-hardness can easily be verified by noticing that the \setcov problem can be reduced to a specific instance of \mincat where each path in ${\mathcal P}$ can be formed using a single catalyst. In this case, in fact, we ask for the
minimum number of catalysts required to cover at least one path of every $\langle s, t \rangle$ pair, which exactly corresponds to what \setcov asks for.
\end{proof}

\vspace{-3mm}
\spara{Algorithm for \mincat.} We design an algorithm to provide effective solutions to \mincat which consists of four steps:
\begin{itemize}
\setlength\itemsep{0.1em}
\item {\bf Step 1.} Mark all $\langle s, t \rangle$ pairs as disconnected.
\item {\bf Step 2.} For all disconnected $\langle s, t \rangle$ pairs, find a path $P$ that connects one of such $\langle s, t \rangle$ pairs, while adding the minimum number of new catalysts to the set of already selected catalysts.
\item {\bf Step 3.} Mark that $\langle s, t \rangle$ pair as connected. Include the catalysts in path $P$ to the set of selected catalysts.
\item {\bf Step 4.} If there is at least one disconnected $\langle s, t \rangle$ pair, go to step 2.
\end{itemize}
\begin{figure}[t]
\centering
\includegraphics[width=.42\columnwidth]{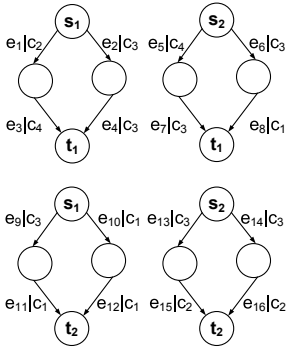}
\vspace{-2mm}
\caption{\scriptsize A demonstration of \topcatmin solution: catalysts that influence the edges are shown in the figure.
However, the corresponding probabilities are not shown. Also, $P(e|c)=0$ for all edge-catalyst combinations not shown.}
\label{fig:min_color}
\vspace{-5mm}
\end{figure}

We report the set of selected catalysts as our minimum set.
If the size of this minimum set is more than $k$, we perform an additional step. From the selected set $C^*$, if a subset $C'$ can be removed, but a path can still be found for all connected  $\langle s, t \rangle$ pairs with the remaining catalysts in $C^* \setminus C'$, then $C'$ is removed from $C^*$. We illustrate our algorithm with an example below.
\begin{example}
As shown in Figure~\ref{fig:min_color}, let us assume that the source set is $S=\{s_1, s_2\}$, and the target set is $T=\{t_1, t_2\}$.
The figure illustrates the top-$2$ most reliable paths for each source-target pair. Assume there is a budget $k=3$ on the
number of output catalysts.
We now apply our algorithm. First, we select the catalyst $c_3$, since this catalyst is sufficient
to have an edge for the $\langle s_1, t_1 \rangle$ pair. Then, we select the catalyst $c_4$ because $\{c_3, c_4\}$ together add an edge for
the  $\langle s_2, t_1 \rangle$ pair. Next, we consider the catalyst $c_1$ in order to have an edge for the pair $\langle s_1, t_2 \rangle$.
At this point, we have already saturated the budget of 3 catalysts: $\{c_1,c_3,c_4\}$, but we are yet to add an edge for the $\langle s_2, t_2 \rangle$ pair.
Thus, we delete $c_4$ from the selected set of catalysts, because this still allows a path for three previously connected source-target pairs.
Finally, we add catalyst $c_2$ to the set. The final set $\{c_1,c_2,c_3\}$ of catalysts allows a path for all source-target pairs.
\end{example}

\vspace{-2mm}
\spara{Time complexity.} In each iteration, we find the path with the smallest number of new catalysts, which requires $\bigO\left(r|S||T|\left(n' + m'\right)\right)$ time. Then,
 we also remove the redundant catalysts, which requires another $\bigO\left(kr|S||T|\left(n' + m'\right)\right)$ time. Since there can be at most $|S||T|$ iterations, overall
 time complexity of our \mincat finding algorithm is $\bigO\left(kr|S|^2|T|^2\left(n' + m'\right)\right)$.

Intuitively, the minimum set finding step ensures that, given large enough budget on catalysts, there will be at least one
path for all $\langle s, t \rangle$ pairs. Therefore, the objective function of the \topcatmin problem will be guaranteed to be larger than zero.
If our
budget has not been exhausted yet and more catalysts can be added, we next apply the \steptwo algorithm as follows. At each iteration,
we find the $\langle s, t \rangle$ pair exhibiting the minimum conditional reliability. We then add a path that maintains the budget, while also maximizing
the marginal gain in reliability for that $\langle s, t \rangle$ pair.
The algorithm terminates when no more paths can be added without exceeding the budget $k$, or all top-$r$ paths for all $\langle s, t \rangle$ pairs have been selected.
\vspace{-4mm}
\subsection{Maximizing connectivity}
\label{sec:connectivity}
In the second variant of the \topcat problem applied to multiple query nodes, we do not distinguish between source and target nodes. All query nodes are considered as peers: the objective of this \topcatconn problem is to find a set of top-$k$ catalysts which maximize the probability that all query nodes are connected in the subgraph induced by edges containing those catalysts.
Therefore, our problem is inspired by the well-established problem of \emph{k-terminal reliability}, whose objective is to \emph{compute} the probability that a given set of nodes remain connected.
An application of \topcatconn problem is
finding a suitable topic list of a thematic scientific event among researchers.
The event would be successful not only when the invitees are experts
on those topics, but also if they can network with each other, that is,
they can find connections (e.g., direct and indirect links formed due to research collaborations)
with other invitees based on those topics \cite{SG10}.
We formally define our problem below.

\begin{problem} [\topcatconn]
Given an uncertain graph $\mathcal{G}=(V,E,C,P)$, a set of query nodes $Q \subset V$, and a positive integer $k$,
find a set $C^*$ of catalysts, with size $k$, that maximizes the probability of nodes in $Q$ being connected only using catalysts in $C^*$:
\vspace{-2mm}
\begin{align}
& C^* = \argmax_{C_1 \subseteq C} \sum_{G \sqsubseteq \mathcal{G}|C_1} [J_G(Q) \times P(G|C_1)] \nonumber & \\
& \text{such that}\quad |C_1|=k. &
\vspace{-2mm}
\end{align}
\vspace{-2mm}
\label{prob:connectivity}
\end{problem}
\vspace{-2mm}
\vspace{-2mm}
In the above statement, $J_G(Q)$ is an indicator function over a possible deterministic graph $G \sqsubseteq \mathcal{G}|C_1$ taking value $1$ if nodes in $Q$ are all connected in $G$, and $0$ otherwise. For simplicity, in directed graphs we consider a weak notion of connectivity, i.e., connectivity disregarding edge-directions.
The extension to strong connectivity is straightforward.

%\vspace{0.8mm}
\spara{Algorithm.}
The \topcatconn problem is a generalization of the \topcat basic problem (Problem~\ref{prob:topk_rel_col_set}).
Thus, it can be immediately be recognized as \NP-hard.

To provide high quality results, we design a two-step algorithm whose outline is similar in spirit to the \ours algorithm proposed for \topcat.
As a main difference, however, since our goal in \topcatconn is to maximize {\em connectivity} among a set of peer nodes,
we ask for the top-$r$ minimum {\em Steiner trees} as a first step of the algorithm (rather than top-$r$ most reliable paths between a single source-target pair).
A Steiner tree for a set $Q$ of nodes in a weighted graph is a tree that spans all nodes of $Q$.
A minimum Steiner tree is a Steiner tree whose sum of edge-weights is the minimum. We first apply
the technique  proposed in \cite{DYWQZL07} to find the top-$r$ minimum Steiner trees from an equivalent edge-weighted, multi-graph
$\mathcal{G''}$. We recall that $\mathcal{G''}$ can be obtained from the input uncertain graph $\mathcal{G}$ by following Algorithm~\ref{alg:reliable_path}.
Next, we iteratively include the Steiner trees in our solution so as to maximize the marginal gain in the probability that nodes in $Q$ are connected, while not exceeding the budget on catalysts.

\spara{Time complexity.} The complexity to find the top-$r$ minimum Steiner trees is $\bigO\left(3^{|Q|}n+2^{|Q|}\left(\left(|Q| + \log n\right)n + e\right)\right)$ \cite{DYWQZL07}.
As for \steptwo, the complexity of our iterative tree inclusion method is $\bigO(r^2(n' +
m')K)$, where, we recall, $K$ is the number of samples used in each {\sf MC} sampling, and $n'$ and $m'$ are the number of nodes and edges in
the subgraph induced by the top-$r$ minimum Steiner trees, respectively.
\begin{table} [tb!]
\begin{scriptsize}
\begin{center}
\vspace{-4mm}
\centering
\begin{tabular} { @{}l@{ }|@{  }c@{ }c@{ }c@{ }c@{} }
%\hline
 &   &   &     & {\textsf{edge probabilities:}} \\
 {\textsf{dataset}} &      {\textsf{nodes}}               &    {\textsf{edges}}                   &      {\textsf{catalysts}}                        & {\textsf{mean, SD, quartiles}} \\ \hline
 {\em DBLP}     &        1\,291\,297  &    3\,561\,816    &     347                               &  0.21, 0.08, \{0.181, 0.181, 0.181\}  \\
 {\em BioMine}  &        1\,045\,414  &    6\,742\,943    &     20                                &  0.27, 0.17, \{0.116, 0.216, 0.363\}\\
 {\em Freebase} &        28\,483\,132 &    46\,708\,421   &     5\,428                            &  0.50, 0.24, \{0.250, 0.500, 0.750\}\\
  \hline
  \end{tabular}
\vspace{-4mm}
\end{center}
\end{scriptsize}
\caption{\scriptsize Characteristics of the uncertain graphs used in the experiments. \label{tab:data}}
\vspace{-5mm}
\end{table} 
\vspace{-4mm}
\section{Experimental evaluation}
\label{sec:experiments}

\vspace{-1mm}
We report empirical results to show accuracy, efficiency, and memory usage of the proposed methods.
We also provide results on information diffusion to demonstrate the applicability of the top-$k$ catalysts identified
by our methods. We report sensitivity analysis by varying all main parameters:
the number of catalysts, reliable paths, query nodes, and the distance between source and target nodes.
The code is implemented in C++ and experiments are performed on a single core of a 100GB, 2.26GHz Xeon server.

\vspace{-3mm}
\subsection{Experimental setup}
\label{sec:setting}

\vspace{-1mm}
\spara{Datasets.} We use three real-world uncertain graphs.

{\em DBLP (http://dblp.uni-trier.de/xml).} We use this well-known collaboration network, downloaded on
August 31, 2016. Each node represents an author, and an edge denotes co-authorship.
Each edge is defined by a set of keywords, that are present within the title of the papers, co-authored by the respective authors.
We selected 347 distinct keywords from all paper titles, e.g., databases, distributed, learning, crowd, verification, etc, based on
frequency and how well they represent various sub-areas of computer science.
We count occurrences of a specific keyword in the titles of the papers co-authored by any two authors.
Edge probabilities are derived from an exponential cdf of mean $\mu=5$ to this count \cite{JLDW11};
hence, if a keyword $c$ appeared $t$ times in the titles of the papers co-authored by the authors $u$ and $v$,
the corresponding probability is $p((u,v)|c))=1 - exp^{-t/5}$. The intuition is that the more the times
$u$ and $v$ co-authored on keyword $c$, the higher the chance (i.e., the probability) that $u$ influences $v$ (and, vice versa)
for that keyword. Therefore, keywords correspond to catalysts for information
cascade.

{\em BioMine (https://www.cs.helsinki.fi/group/biomine).} This is the database of the {\sf BIOMINE} project  \cite{Eronen2012}.
The graph is constructed by integrating cross-references from several biological databases. Nodes
represent biological concepts such as genes, proteins, etc., and edges denote
real-world phenomena between two nodes, e.g., a gene ``codes'' for a protein.
In our setting these phenomena correspond to catalysts. Edge probabilities, which quantify the existence of
a phenomenon between the two endpoints of that edge, were determined in \cite{Eronen2012} as a combination of three criteria: relevance (i.e., relative importance of that relationship type), informativeness (e.g., degrees of the nodes adjacent to that edge), confidence on the existence of a specific relationship (e.g., conformity with the biological {\sf STRING} database).
%, specificity of the source article).

{\em Freebase (http://www. freebase.com).}  This is a knowledge graph, where nodes are named entities (e.g.,
Google) or abstract concepts (e.g., Asian people), while edges represent relationships among those
entities (Jerry Yang is the ``founder'' of Yahoo!).
Relationships corresponds to catalysts.
We use the probabilistic version of the graph  \cite{CGBY14}.

\vspace{-1mm}
\spara{Query selection.}
For each set of experiments, we select 500 different queries. If we do not impose any distance constraint between the source and the target, both of them are picked
uniformly at random. When we would like to maintain a maximum pairwise distance $d$ from the source to the target, we first select the source uniformly at random. Then, out of
all nodes that are within $d$-hops from it, one node is selected uniformly at random as the target.
All reported results are averaged over 500 such queries.

\vspace{-1mm}
\spara{Competing methods.}
We compare the proposed \ours method (Algorithm \ref{alg:ours}) to the two baselines, \bone and \btwo, discussed in Sections \ref{sec:ind}--\ref{sec:hill_climb}.
For the sake of brevity, in the remainder of this section we refer to the proposed method as
\oursshort, and to the \bone baseline as \boneshort.

\vspace{-1mm}
\spara{Reliability estimation.}
Our proposed method and the baselines need a subroutine that estimates conditional reliability for given source node(s), target node(s), and number of catalysts.
To this end, we employ the well-established {\sf Monte Carlo}-sampling method.
In particular, to improve efficiency, we combine {\sf MC} sampling with a breadth first search from the source node (set) \cite{JLDW11}, meaning
that the coin for establishing if an edge should be included in the current sample is flipped only upon request. This avoids to flip coins for edges
in parts of the graph that are not reached with the current breadth first search, thus increasing the chance of an early termination.
In the experiments, we found that {\sf MC} sampling converges at around $K = 1000$ samples in our datasets. This is roughly the same number observed
in the literature \cite{JLDW11,PBGK10} for these datasets. Hence, we set $K = 1000$ in all sets of experiments.

\vspace{-4mm}
\subsection{Single-source single-target}
\vspace{-1mm}
\spara{Experiments over different datasets.}
In Table~\ref{tab:performance}, we show conditional reliability and running time of all competitors for top-5 output catalysts.
For our \oursshort, we use top-20 most reliable paths with {\em Freebase} and {\em BioMine} and top-50 most reliable paths over {\em DBLP},
as we observe that, for finding the top-5 catalysts, increasing the number of paths beyond 20 ({\em Freebase} and {\em BioMine}) and 50 ({\em DBLP}) does not significantly increase the quality
in respective datasets. Results with varying the number of most reliable paths, and its dependence on varying number of top-$k$ catalysts, will
be reported shortly.

Conditional reliability illustrates the quality of the top-$k$ catalysts found: the higher the reliability, the better the quality.
The proposed \oursshort achieves the best quality results on all our datasets.
\begin{table}[tb!]
%\vspace{-2mm}
\begin{scriptsize}
\begin{center}
\begin{tabular} { c|ccc|ccc}
%\hline
\multicolumn{1}{c}{} & \multicolumn{3}{c}{\textsf{conditional reliability}} &  \multicolumn{3}{c}{\textsf{running time (sec)}} \\
\textsf{dataset}    & \boneshort & \btwoshort & \oursshort & \boneshort & \btwoshort & \oursshort \\ \hline
{\em Freebase}       &     0.15      &    0.15       &         {\bf 0.17}&     1.38     &    43          &        {\bf 0.02} \\
{\em BioMine}        &     0.18      &    0.43       &         {\bf 0.59}&     1220     &    26217       &        {\bf 5.27}\\
{\em DBLP}           &     0.11      &    0.26       &         {\bf 0.28}&     85.97    &    36519       &        {\bf 1.07} \\ \hline
\end{tabular}
\end{center}
\end{scriptsize}
\vspace{-4mm}
\caption{\scriptsize Reliability and efficiency over different datasets. Single source-target pair, top-5 catalysts.\label{tab:performance}}
\vspace{-3mm}
\end{table}

\begin{table}[tb!]
%\vspace{-2mm}
\begin{scriptsize}
\begin{center}
\begin{tabular} { c|ccc|ccc}
%\hline
\multicolumn{1}{c}{} & \multicolumn{3}{c}{\textsf{conditional reliability}} &  \multicolumn{3}{c}{\textsf{running time (sec)}} \\
$k$    & \boneshort & \btwoshort & \oursshort & \boneshort & \btwoshort & \oursshort \\ \hline
 5                  &      0.18     &    0.43        &  {\bf 0.59}       &    1220       &  26217         &     {\bf 5.27}              \\
 8                  &      0.18     &    0.49        &  {\bf 0.59}       &    2210       &  67158         &     {\bf 7.05}              \\
10                  &      0.18     &    0.50        &  {\bf 0.60}       &    2290       &  131674        &     {\bf 7.37}              \\
12                  &      0.23     &    0.53        &  {\bf 0.62}       &    2305       &  161265        &     {\bf 7.98}              \\
15                  &      0.34     &    0.53        &  {\bf 0.63}       &    2365       &  217496        &     {\bf 8.30}              \\ \hline
  \end{tabular}
\end{center}
\end{scriptsize}
\vspace{-4mm}
\caption{\scriptsize Reliability and efficiency with varying number $k$ of output catalysts. Single source-target pair, {\em BioMine} dataset. \label{tab:performance_top-k}}
\vspace{-2mm}
\end{table}

\begin{table}[tb!]
%\vspace{-4mm}
\begin{scriptsize}
\begin{center}
\begin{tabular} { c|ccc|ccc}
%\hline
\multicolumn{1}{c}{distance} & \multicolumn{3}{c}{\textsf{conditional reliability}} &  \multicolumn{3}{c}{\textsf{running time (sec)}} \\
\textsf{(\# hops)}    & \boneshort & \btwoshort & \textsf{\oursshort} & \boneshort & \btwoshort & \oursshort \\ \hline
2                 &   0.45        &      0.75 &       {\bf 0.83}   &       346  & 9798      &      {\bf 4.90}        \\
4                 &   0.08        &      0.38 &       {\bf 0.64}   &       406  & 23140     &      {\bf 5.37}        \\
6                 &   0.02        &      0.17 &       {\bf 0.30}   &       548  & 29135     &      {\bf 5.58}        \\ \hline
  \end{tabular}
\end{center}
\end{scriptsize}
\vspace{-4mm}
\caption{\scriptsize Reliability and efficiency with varying distance between the source and the target. Single source-target pair, {\em BioMine} dataset, top-5 catalysts.\label{tab:performance_distance}}
\vspace{-5mm}
\end{table}

Concerning running time, we observe that \oursshort is 2-3 orders of magnitude faster than \boneshort, and 3-4 orders faster than \btwoshort. This confirms that performing {\sf MC} sampling on a significantly reduced version of the input graph leads to significant benefits in terms of efficiency, without affecting accuracy.
Surprisingly, \btwoshort is orders of magnitude slower than \boneshort.
The reason is the following.
Although only a factor $k$ separates \boneshort
from \btwoshort based on our complexity analysis, what happens in practice is that \boneshort benefits from {\sf MC}-sampling's early termination much more than \btwoshort, as \boneshort considers each catalyst individually, while \btwoshort considers a set of catalysts. One may also note that the running times over {\em BioMine} is higher than that over {\em Freebase}. Although {\em Freebase} has more nodes and edges, the graph is sparse compared to {\em BioMine}. Therefore, a breadth first search in \emph{BioMine} often traverses more nodes, thus increasing its processing time.

\spara{Varying number of catalysts.}
We show results with varying the number $k$ of output catalysts in  Table~\ref{tab:performance_top-k} and Figure~\ref{fig:r_k_dblp}.
Similar trends have been observed in all datasets, thus, for brevity, we report results on {\em BioMine} (Table~\ref{tab:performance_top-k}, $k$ varies from $5$ to $15$)
and on {\em DBLP} (Figure~\ref{fig:r_k_dblp}, $k$ varies from $10$ to $100$). As expected, conditional reliability and running time increase with more catalysts.
Moreover, as shown in Table~\ref{tab:performance_top-k}, our \oursshort remains more accurate and faster than both baselines for all $k$.

\spara{Varying distance from the source to the target.}
Table~\ref{tab:performance_distance} reports on results with varying the distance between the source and the target.
Keeping fixed the number of output catalysts, as expected, the reliability achieved by all three methods decreases with larger distance from the source to the target.
However, we observe that the reliability drops sharply for \boneshort. This is because with increasing distance, it becomes less likely that there would be a path due to only one catalyst from the source to the target.
We also note that the reliability decreases more in \btwoshort than in the proposed \oursshort.
This is due to the cold-start problem of \btwoshort: It is more likely for \btwoshort to make mistaken choices in the initial steps if the source and the target are connected by longer paths.

\spara{Varying number of most reliable paths.}
We also test \oursshort for different values of the number $r$ of most reliable paths discovered in the first step of the algorithm.
We report these results for {\em BioMine} and {\em Freebase}
in Table~\ref{tab:performance_top-r}, and for {\em DBLP} in Figures~\ref{fig:r_k_dblp}, \ref{fig:sufficient_dblp_r}.
For {\em BioMine} and {\em Freebase} datasets, we fix the number $k$ of output catalysts as $5$,
and we observe that while increasing the number of paths, the reliability initially increases, then saturates at a certain value of $r$ (e.g.,
$r=20$ for {\em BioMine} and $r=15$ for {\em Freebase}).
This behavior is expected, since the subsequent paths have very small reliability. Hence, including them does not significantly
increase the quality of the solution found so far. On the other hand, the running time increases almost linearly when more
top-$r$ paths are considered.
\begin{table}[tb!]
%\vspace{-4mm}
\begin{scriptsize}
\begin{center}
\begin{tabular} { c|c|c||c|c}
%\hline
\multicolumn{1}{c}{} & \multicolumn{1}{c}{\textsf{cond. reliability}} &  \multicolumn{1}{c}{\textsf{running time (sec)}} & \multicolumn{1}{c}{\textsf{cond. reliability}} &  \multicolumn{1}{c}{\textsf{running time (sec)}}\\
     & \textsf{\oursshort} &  \textsf{\oursshort} & \textsf{\oursshort} &  \textsf{\oursshort}\\
$r$  & BioMine             &     BioMine          &   Freebase          &   Freebase         \\ \hline\hline
1                    &   0.27             &    4.29    & 0.12       & 0.0004 \\
2                    &   0.29             &    4.26    & 0.12       & 0.002  \\
3                    &   0.31             &    4.30    & 0.13       & 0.002  \\
4                    &   0.31             &    4.30    & 0.14       & 0.003  \\
5                    &   0.31             &    4.31    & 0.14       & 0.004  \\
10                   &   0.32             &    4.31    & 0.16       & 0.008  \\
15                   &   0.32             &    4.37    & {\bf 0.17} & 0.013  \\
20                   &   {\bf 0.33}       &    5.26    & 0.17       & 0.018  \\
30                   &   0.33             &    5.29    & 0.17       & 0.020  \\
50                   &   0.33             &    5.38    & 0.17       & 0.035  \\
100                  &   0.33             &    5.70    & 0.17       & 0.081  \\ \hline
  \end{tabular}
\end{center}
\end{scriptsize}
\vspace{-4mm}
\caption{\scriptsize Reliability and efficiency with varying number $r$ of most reliable paths in the proposed \oursshort. Single source-target pair, top-5 catalysts.\label{tab:performance_top-r}}
\vspace{-5mm}
\end{table}
\begin{figure}[tb!]
\centering
\subfigure [\scriptsize Conditional Reliability] {
\includegraphics[scale=0.28]{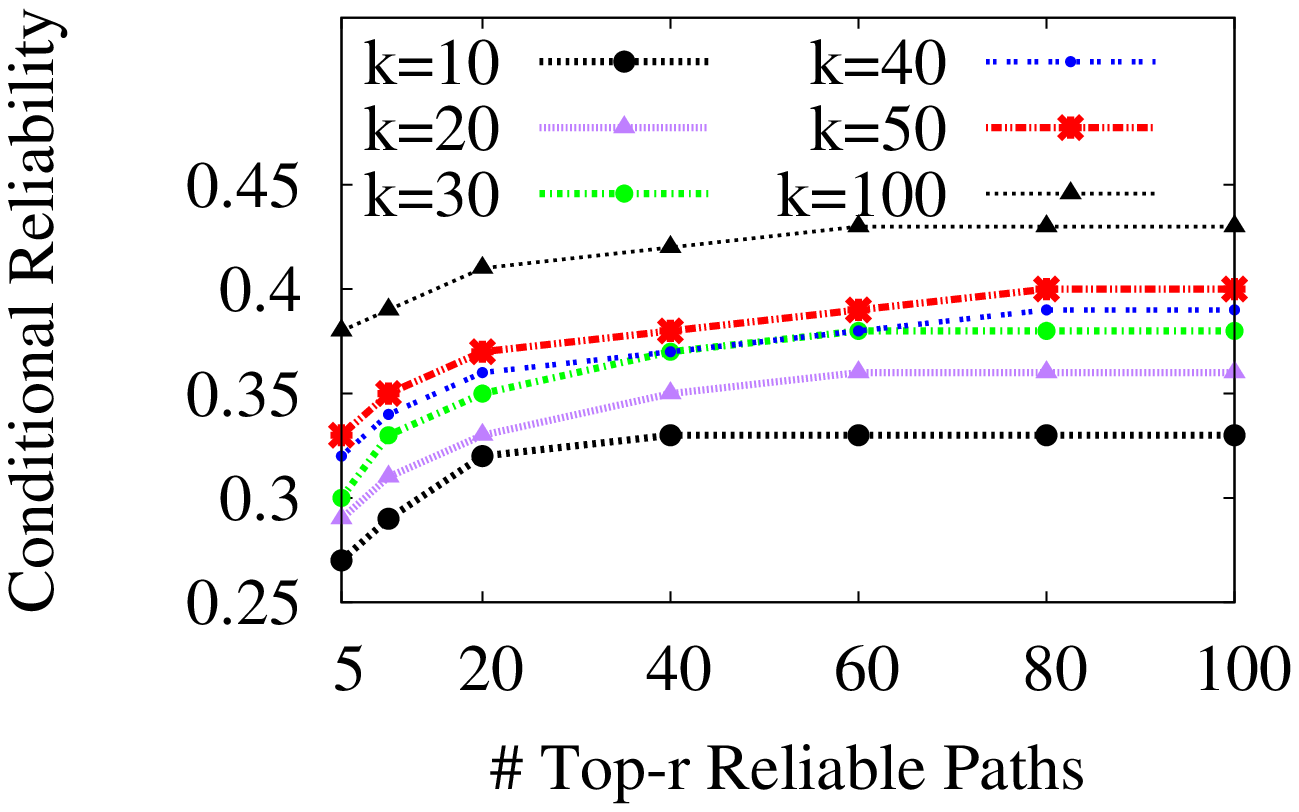}
\label{fig:r_k_dblp_rel}
}
\subfigure [\scriptsize Running Time] {
\includegraphics[scale=0.28]{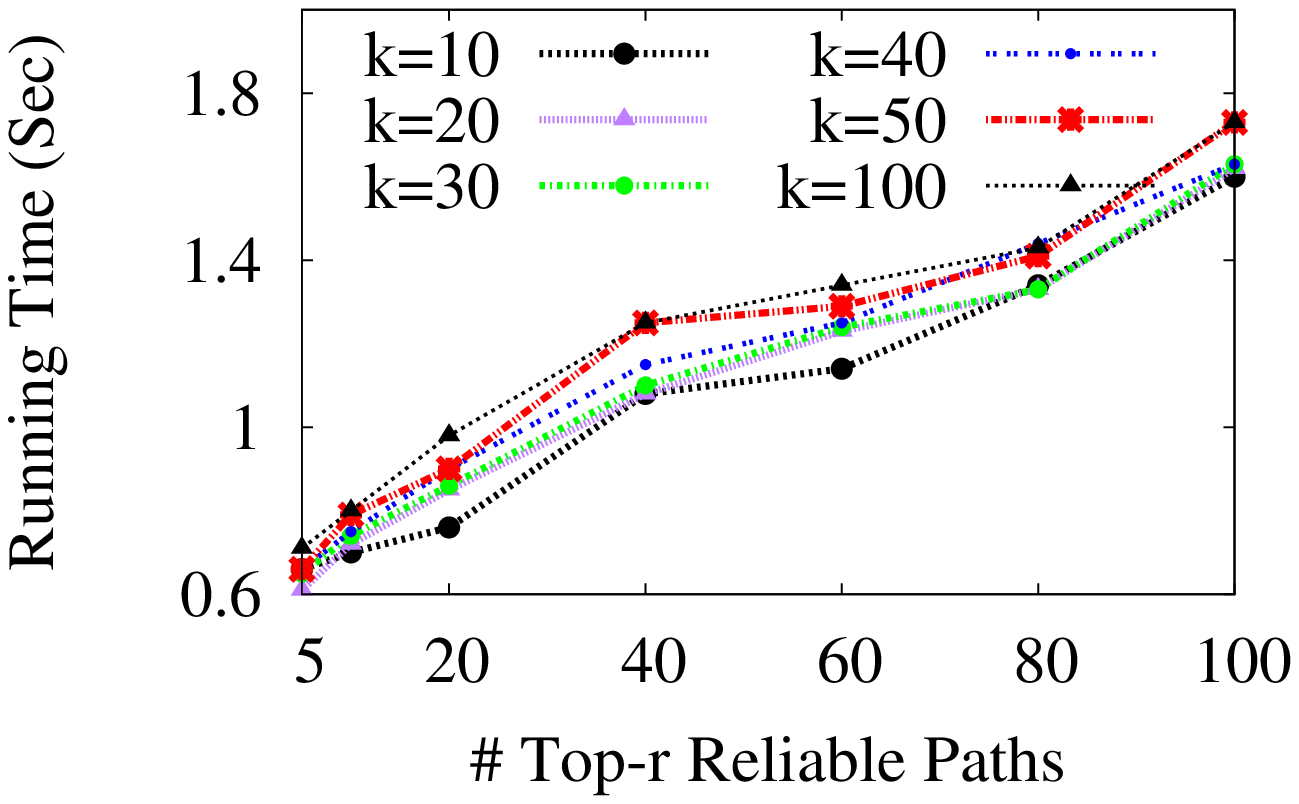}
\label{fig:r_k_dblp_time}
}
\vspace{-4mm}
\caption{\scriptsize Reliability and efficiency with varying number $r$ of most reliable paths in the proposed \oursshort. Single source-target pair, number of top-$k$ catalysts vary from $k$=10 to $k$=100, {\em DBLP}.}
\label{fig:r_k_dblp}
\vspace{-5mm}
\end{figure}
\begin{figure}[tb!]
\centering
\includegraphics[scale=0.28]{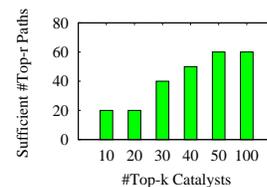}
\vspace{-4mm}
\caption{\scriptsize Sufficient number of top-$r$ reliable paths to find the top-$k$ catalysts by \oursshort,
number of top-$k$ catalysts vary from $k$=10 to $k$=100, {\em DBLP}.}
\label{fig:sufficient_dblp_r}
\vspace{-5mm}
\end{figure}
\begin{table}[t!]
\begin{varwidth}[b]{0.6\linewidth}
\centering
\vspace{1.8mm}
\begin{scriptsize}
\begin{tabular} { l|l}
%\hline
\multicolumn{1}{c}{\textsf{Datasets}} & \multicolumn{1}{c}{\textsf{Memory Usage}} \\ \hline\hline
{\em DBLP} (1.3M, 3.6M) &  1.9 GB   \\
{\em BioMine} (1.0M, 6.7M) & 1.8 GB \\
{\em Freebase} (28.5M, 46.7M) & 16.0 GB \\ \hline
\end{tabular}
\end{scriptsize}
\vspace{-2mm}
\caption{\scriptsize Memory usage for \oursshort \label{tab:mem}}
\vspace{-3mm}
\end{varwidth}%
\hfill
\begin{minipage}[b]{0.37\linewidth}
\centering
\vspace{-4mm}
\includegraphics[scale=0.25]{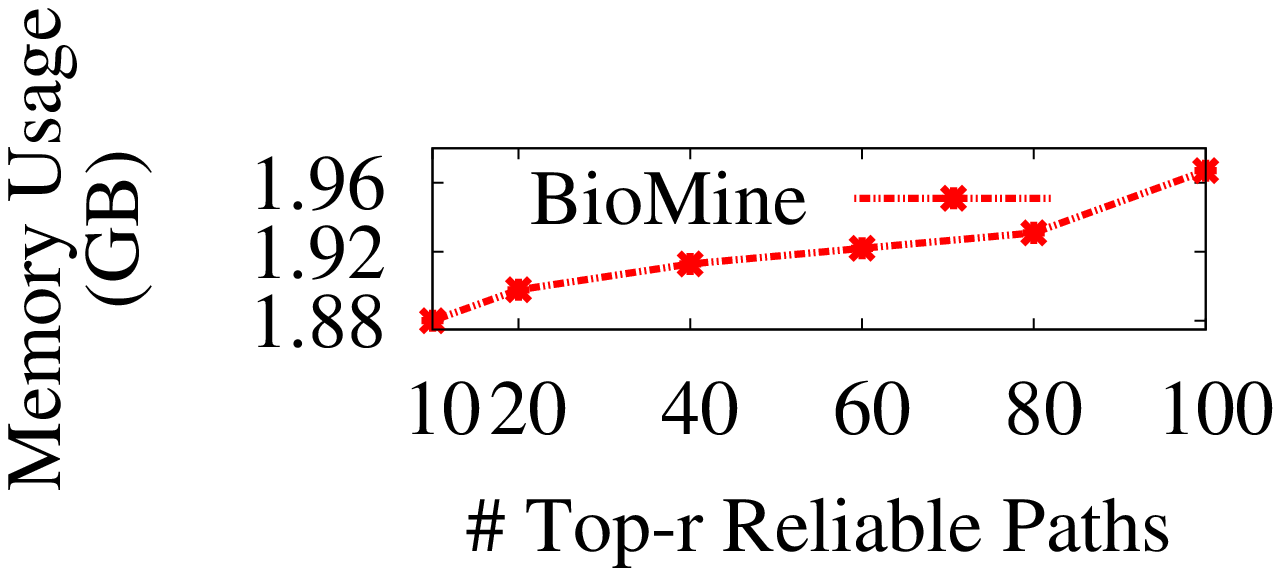}
\vspace{-7mm}
\captionof{figure}{\scriptsize Varying \#rel. paths}
\label{fig:mem_r}
\vspace{-2mm}
\end{minipage}
\vspace{-5mm}
\end{table}

A similar behavior is observed in the {\em DBLP} dataset.
Here we additionally vary the number $k$ of output catalysts
from $10$ to $100$ (Figure~\ref{fig:r_k_dblp}), and we find that, as $k$ increases,  a larger set of reliable paths need to be considered to make accuracy stabilize.
For instance, for $k=10$,
about $r = 20$ reliable paths suffice to observe no more tangible accuracy improvement.
On the other hand, for $k=100$, $r = 60$ paths are required (Figure~\ref{fig:sufficient_dblp_r}).
Once again, this behavior is expected: the larger the number $k$ of catalysts to be output, the larger the subgraph connecting source to target to be explored, and, hence, the larger the number of paths to be considered so as to satisfactorily cover that subgraph.
%as we increase $k$, a larger set of reliable paths needs to be considered to obtain high-quality
%output $k$ catalysts. For example, in Figure~\ref{fig:r_k_dblp_rel}, when the number of catalysts $k=10$,
%about top-$20$ reliable paths are sufficient to achieve a good solution. On the contrary, when $k=100$, one may require
%up to top-$60$ reliable paths. Once again, this behavior is expected --- when one increases the budget $k$ on the number of output catalysts,
%one requires to consider a larger subgraph between the source and the target.
%These empirical results help us decide the optimal value of $r$, i.e., the number of top-$r$ reliable paths required for a specific value of $k$, i.e., the number of top-$k$ catalysts.

\spara{Memory usage.} We report the memory usage of \oursshort in Table~\ref{tab:mem}. This
is dominated by the space required for the graph in the main memory. The top-$r$ reliable paths selected by our algorithm
consumes only a few tens of megabytes. Moreover, the memory consumption increases linearly with the number of reliable paths
selected (Figure~\ref{fig:mem_r}).
\begin{figure}[tb!]
\centering
\subfigure [\scriptsize Conditional Reliability] {
\includegraphics[scale=0.23]{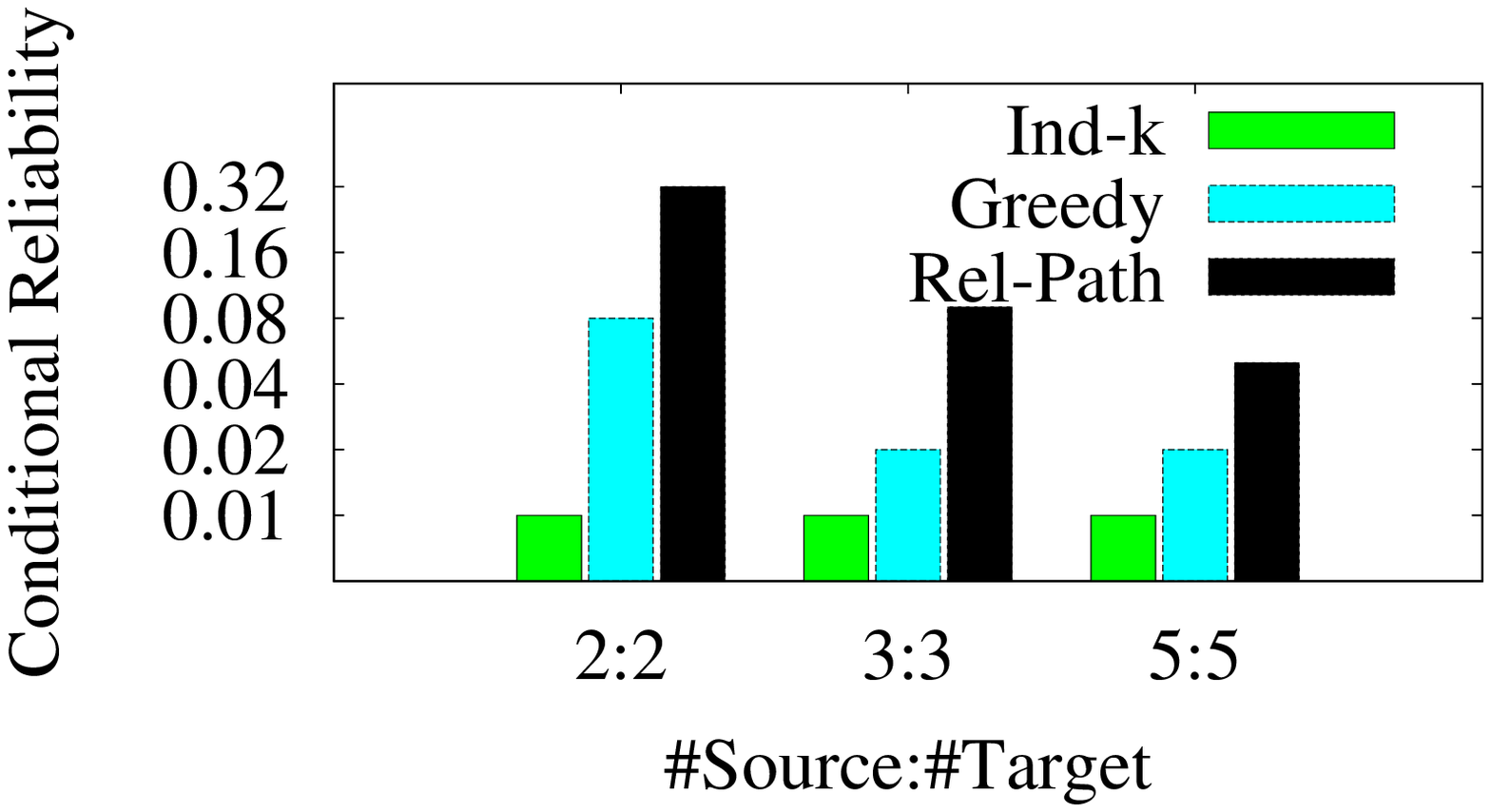}
\vspace{-3mm}
\label{fig:min_freebase_rel}
}
\subfigure [\scriptsize Running Time] {
\includegraphics[scale=0.23]{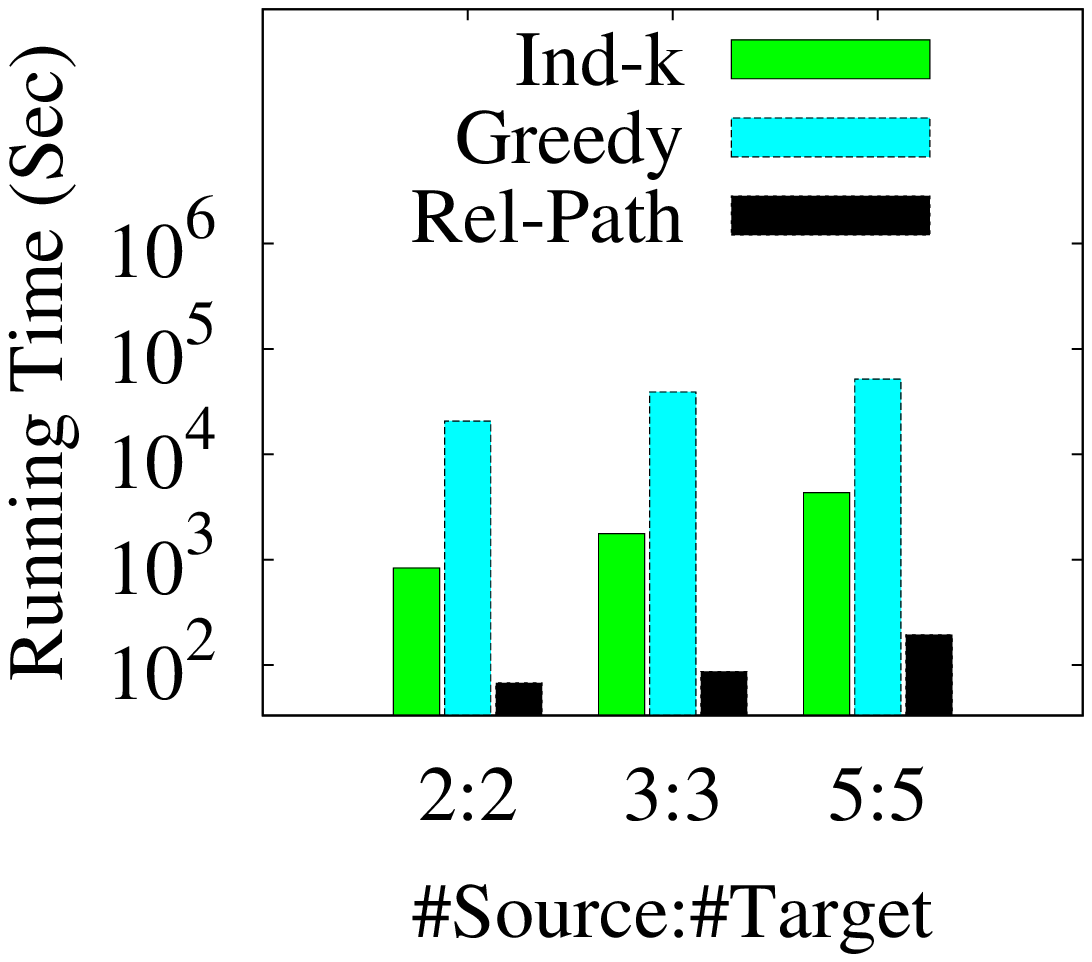}
\vspace{-3mm}
\label{fig:min_freebase_time}
}
\vspace{-4mm}
\caption{\scriptsize Reliability and efficiency for multiple source-target pairs: {\em Freebase}, top-5 catalysts, aggregate function = minimum.}
\label{fig:rel_freebase}
\vspace{-3mm}
\end{figure}
\begin{figure}[tb!]
\vspace{-4mm}
\centering
\subfigure [\scriptsize Reliability]{
\includegraphics[scale=0.23]{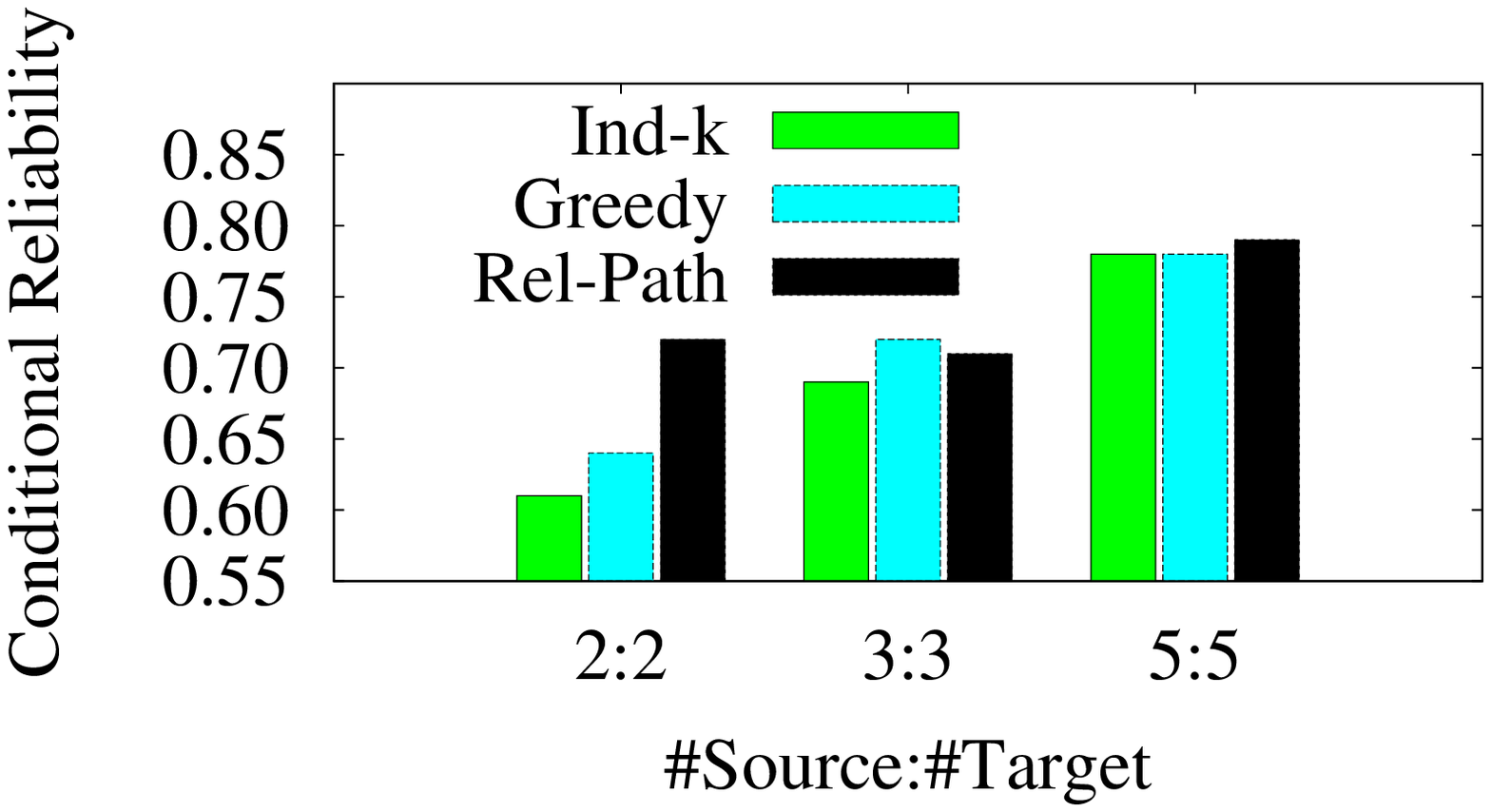}
\vspace{-3mm}
\label{fig:max_freebase_rel}
}
\subfigure [\scriptsize Running Time] {
\includegraphics[scale=0.23]{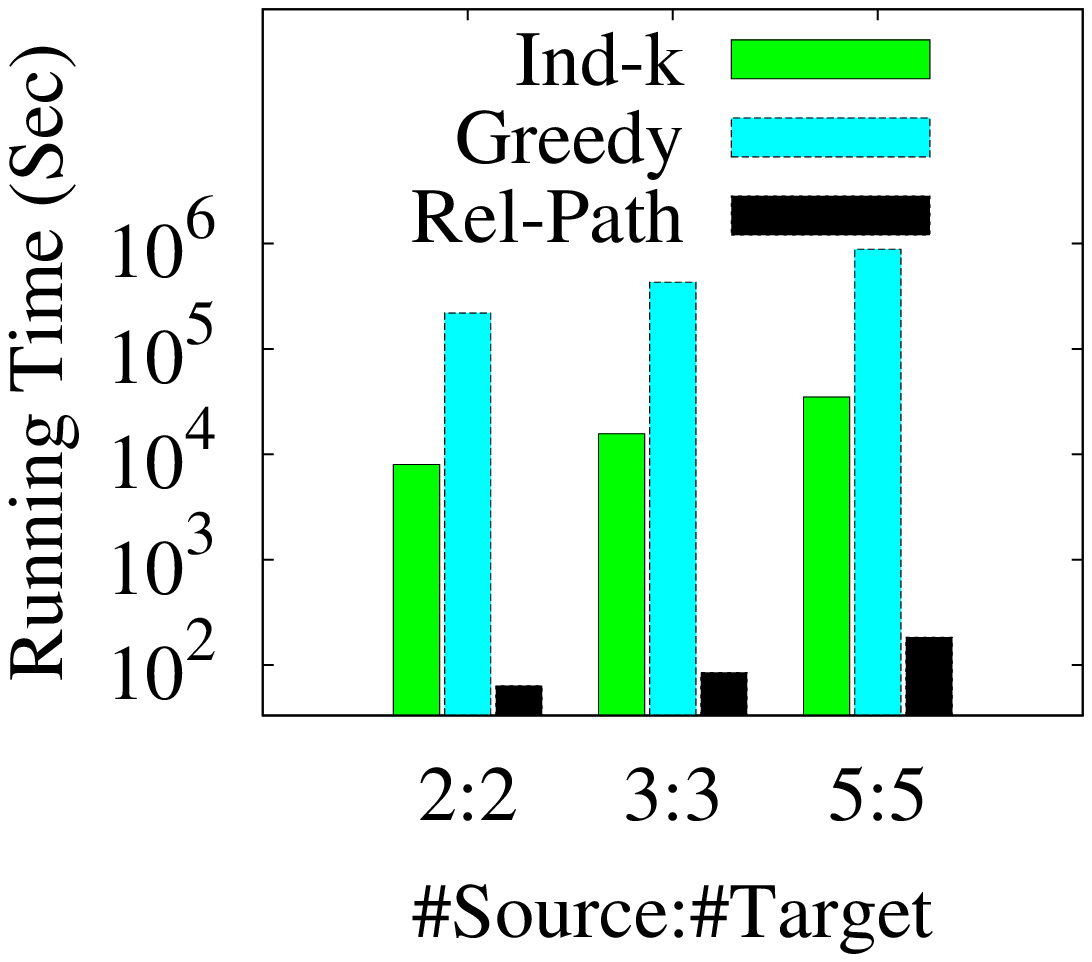}
\vspace{-3mm}
\label{fig:max_freebase_time}
}
\vspace{-4mm}
\caption{\scriptsize Reliability and efficiency for multiple source-target pairs: {\em Freebase}, top-5 catalysts, aggregate function = maximum}
\label{fig:max_freebase}
\vspace{-3mm}
\end{figure}
\begin{figure}[tb!]
\vspace{-4mm}
\centering
\subfigure [\scriptsize Reliability]{
\includegraphics[scale=0.23]{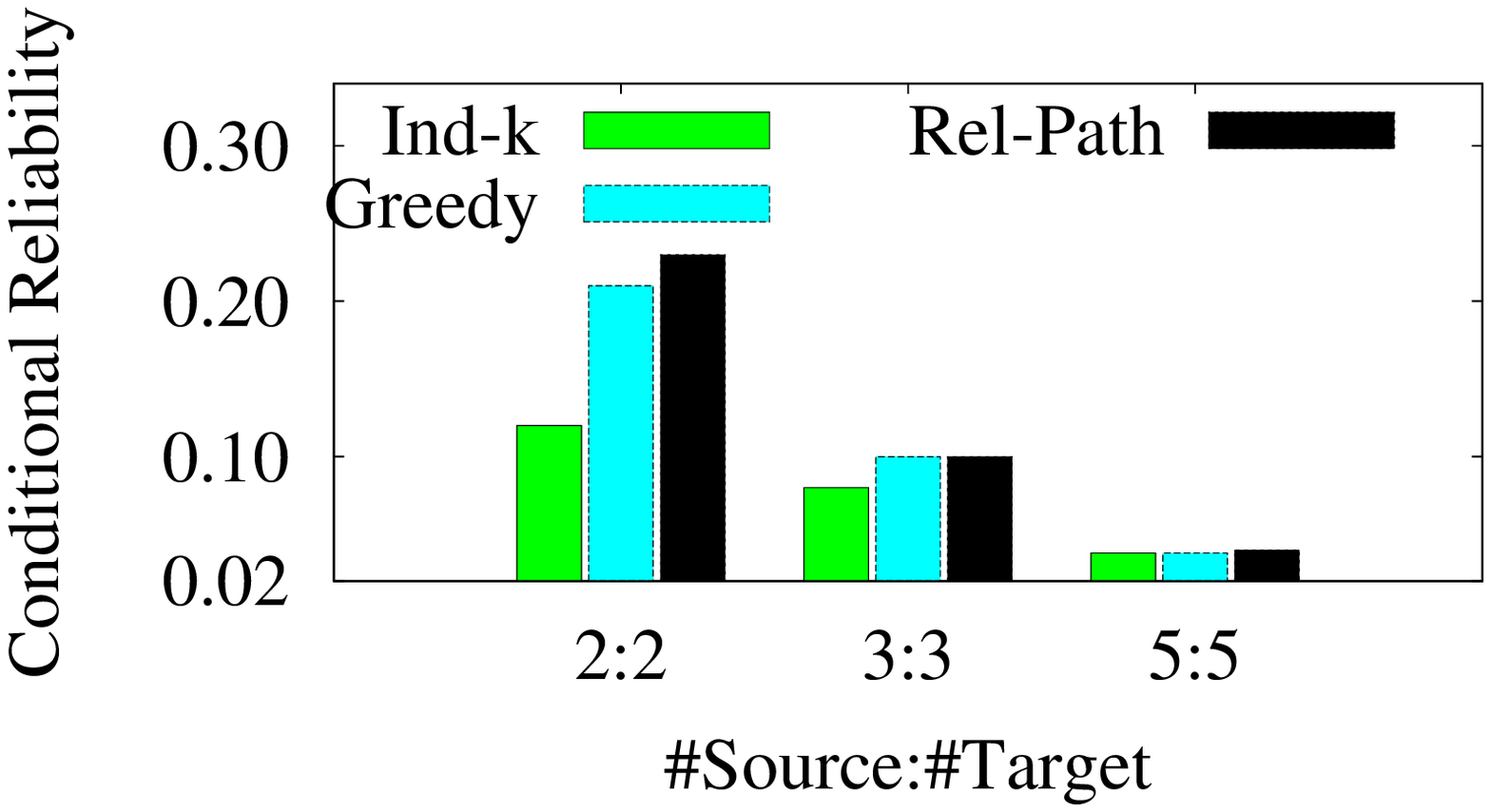}
\vspace{-3mm}
\label{fig:avg_dblp_rel}
}
\subfigure [\scriptsize Running Time] {
\includegraphics[scale=0.23]{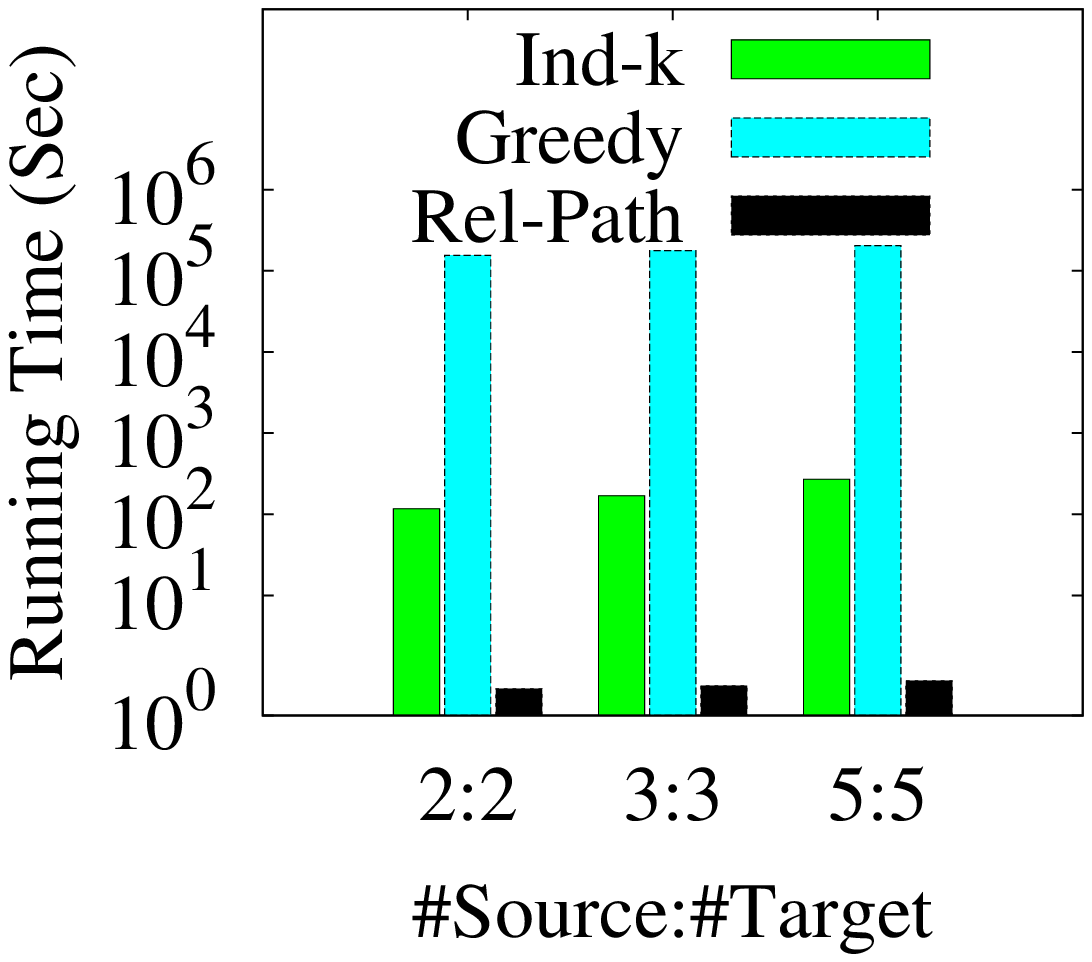}
\vspace{-3mm}
\label{fig:avg_dblp_time}
}
\vspace{-4mm}
\caption{\scriptsize Reliability and efficiency for multiple source-target pairs: {\em DBLP}, top-10 catalysts, aggregate function = average}
\label{fig:avg_dblp}
\vspace{-4mm}
\end{figure}
\begin{figure}[tb!]
%\vspace{-1mm}
\centering
\subfigure [\scriptsize {\em BioMine}, Avg. Aggregate]{
\includegraphics[scale=0.23]{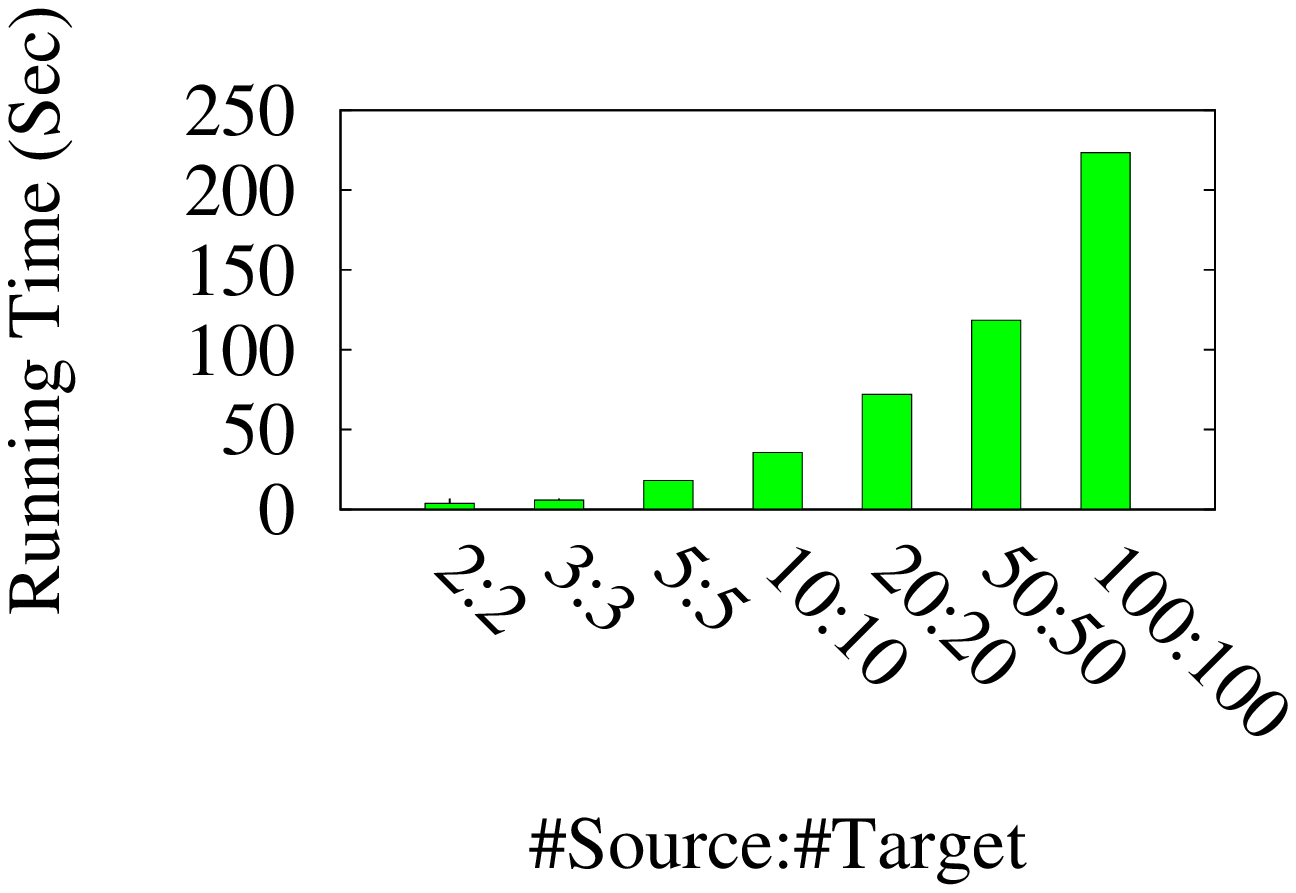}
\vspace{-3mm}
\label{fig:scalability_biomine}
}
\subfigure [\scriptsize {\em Freebase}, Min. Aggregate] {
\includegraphics[scale=0.23]{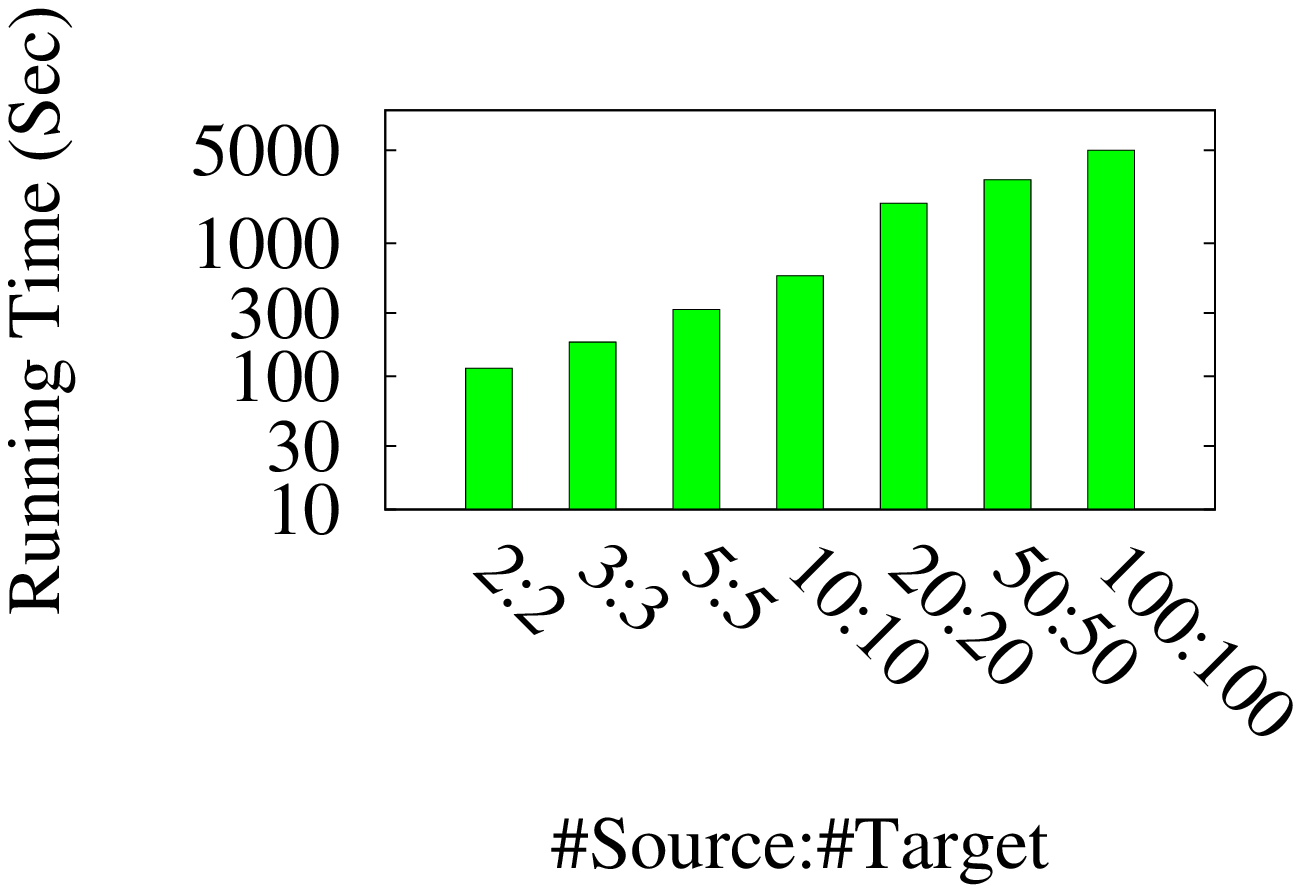}
\vspace{-3mm}
\label{fig:scalability_freebase}
}
\vspace{-4mm}
\caption{\scriptsize Scalability with many sources-targets for \oursshort: top-20 catalysts}
\label{fig:scalability}
\vspace{-6mm}
\end{figure}

\vspace{-4mm}
\subsection{Multiple-sources multiple-targets}

\vspace{-2mm}
\spara{Aggregate functions.}
We perform experiments to evaluate the reliability and efficiency of our methods that maximize an aggregate function over conditional
reliabilities for many source-target pairs. We consider {\sf Minimum} aggregate function, and vary the number of source and target nodes from 2 to 5. In these experiments, we fix the maximum distance between any source-target pair as 4. We also ensure that the same node is not included both in source and target sets.

We show the performance of our algorithms over {\em Freebase} (Figure~\ref{fig:rel_freebase}).
Similar to queries with single source-target pairs, \oursshort outperforms \boneshort and \btwoshort both in terms of efficiency and conditional reliability. Particularly,
due to the presence of multiple source-target pairs, running time differences scale up, and \oursshort is at least
four orders of magnitude faster than the baselines.

We find that with more source-target pairs, the minimum reliability achieved
decreases (Figures~\ref{fig:min_freebase_rel}). This can be explained as follows.
As we keep the number of top-$k$ catalysts fixed at $k=5$, with more source and target nodes, the likelihood of getting one source-target pair with small reliability attained by those top-$k$ catalysts
increases.

\vspace{-0.6mm}
\spara{Different aggregate functions and datasets.}
We demonstrate how our aggregate functions perform over {\em Freebase} and {\em DBLP}, in Figures~\ref{fig:max_freebase} and \ref{fig:avg_dblp},
respectively. Due to common trends, we only show {\sf Maximum} over {\em Freebase} and {\sf Average} over {\em DBLP}. We find that {\sf Rel-Path} results in better reliability
compared to {\sf Greedy} over all experiments. Their difference minimizes in both datasets with more source-target
pairs, which is due to the fact that we keep the number of top-$k$ catalysts fixed at $k=5$ (for {\em Freebase}) and at $k=10$ (for {\em DBLP}). As before,
{\sf Rel-Path} is at least four to five orders of magnitude faster than {\sf Greedy} in all scenarios. In particular, {\sf Greedy}
requires about $10^5$ seconds to answer a single query, which makes almost infeasible to apply this baseline technique in any real-world
online application.

It is interesting that with more source and target pairs, the maximum reliability increases (Figure~\ref{fig:max_freebase}), but the
average reliability decreases (Figure \ref{fig:avg_dblp}). This is expected since with more
source-target pairs, the chance of getting one pair with higher reliability also increases, thereby improving the maximum reliability.
On the contrary, as we consider more source-target pairs while keeping the total number of catalysts same, the average
reliability naturally decreases.

\vspace{-0.6mm}
\spara{Scalability with many sources and targets.}
We demonstrate scalability of our \ours algorithm with multiple source and target nodes (up to 100$\times$100=10K source-target pairs
and top-$20$ catalysts) in Figure~\ref{fig:scalability}. We observe that the running time of \ours increases almost linearly with the number of source-target pairs.
Note that we do not report running times of the \boneshort and \btwoshort baselines, as they do not scale
beyond a small number of sources and targets, as shown earlier in Figures~\ref{fig:rel_freebase}, \ref{fig:max_freebase},
and \ref{fig:avg_dblp}.
\begin{table}[tb!]
\vspace{-1mm}
\begin{scriptsize}
\begin{center}
%\vspace{-3mm}
\begin{tabular} { c|ccc|ccc}
%\hline
\multicolumn{1}{c}{} & \multicolumn{3}{c}{\textsf{Connectivity}} &  \multicolumn{3}{c}{\textsf{Running Time (Sec)}} \\
\textsf{Datasets}    & \textsf{Ind-k}& \textsf{Greedy}& \textsf{Rel-Path} & \textsf{Ind-k}& \textsf{Greedy}& \textsf{Rel-Path} \\ \hline\hline
{\em Freebase}       &     0.01      &    {\bf 0.10} &         {\bf 0.10}&     1\,908 &    13\,175   &        {\bf 80} \\
{\em BioMine}        &     0.29      &    0.47       &         {\bf 0.71}&     893      &    37\,992       &        {\bf 310}\\
{\em DBLP}           &     0.30      &    0.33       &         {\bf 0.35}&     306      &    116\,340       &        {\bf 85} \\ \hline
\end{tabular}
\end{center}
\end{scriptsize}
\vspace{-4mm}
\caption{\scriptsize Connectivity query with 4 nodes, top-5 catalysts}
\label{tab:connectivity}
\vspace{-5mm}
\end{table}

\vspace{-0.6mm}
\spara{Connectivity maximization.}
We illustrate the performance of our algorithms that maximize connectivity (defined in Section~\ref{sec:connectivity}) across multiple query nodes. For these experiments, we select
4 query nodes with maximum pairwise distance between any two nodes fixed at 2. We compare the connectivity
attained by top-$5$ catalysts in Table~\ref{tab:connectivity}. It can be observed that \btwoshort and \oursshort perform equally well in {\em Freebase},
whereas \oursshort results in higher connectivity over {\em BioMine} and {\em DBLP}. We further analyze the top-20 Steiner trees retrieved in {\sf BioMine}, and find that each of these
Steiner trees require 3$\sim$5 distinct catalysts. Therefore, in this dataset, \btwoshort makes more mistakes at initial stages. Because of
the complexity of the top-20 Steiner tree finding algorithm, \oursshort requires more running time in these experiments. However, \oursshort is still significantly
faster that the other two baselines over all our datasets.

\vspace{-4mm}
\subsection{Application in information cascade}
\label{sec:inf_diffusion}
\vspace{-2mm}
Here we showcase our top-$k$ catalysts problem in the context of information diffusion
over social networks. We present our results over the {\em DBLP} dataset.

We select top-$k$ catalysts (i.e., keywords) according to the {\bf Average} aggregate function for multiple sources and targets (see Section 5).
As discussed earlier, if $u$ and $v$ co-authored more on keyword $c$, the higher is the chance (i.e., the probability) that $u$ influences $v$ (and, vice versa) for that keyword. Therefore, keywords correspond to catalysts for information cascade.
\emph{The ultimate goal of this application is to show that the catalysts selected by our method effectively accomplish the task of maximizing the expected spread of information between the source nodes and the target nodes}.
To this purpose, we measure the expected spread achieved by the top-$k$ catalysts selected by our method, and compare it to the expected spread achieved by  (\emph{i}) $k$ random catalysts, and (\emph{ii}) \emph{all} catalysts.
\begin{figure}[tb!]
\vspace{-1mm}
\centering
\subfigure [\scriptsize Information Spread]{
\includegraphics[scale=0.2]{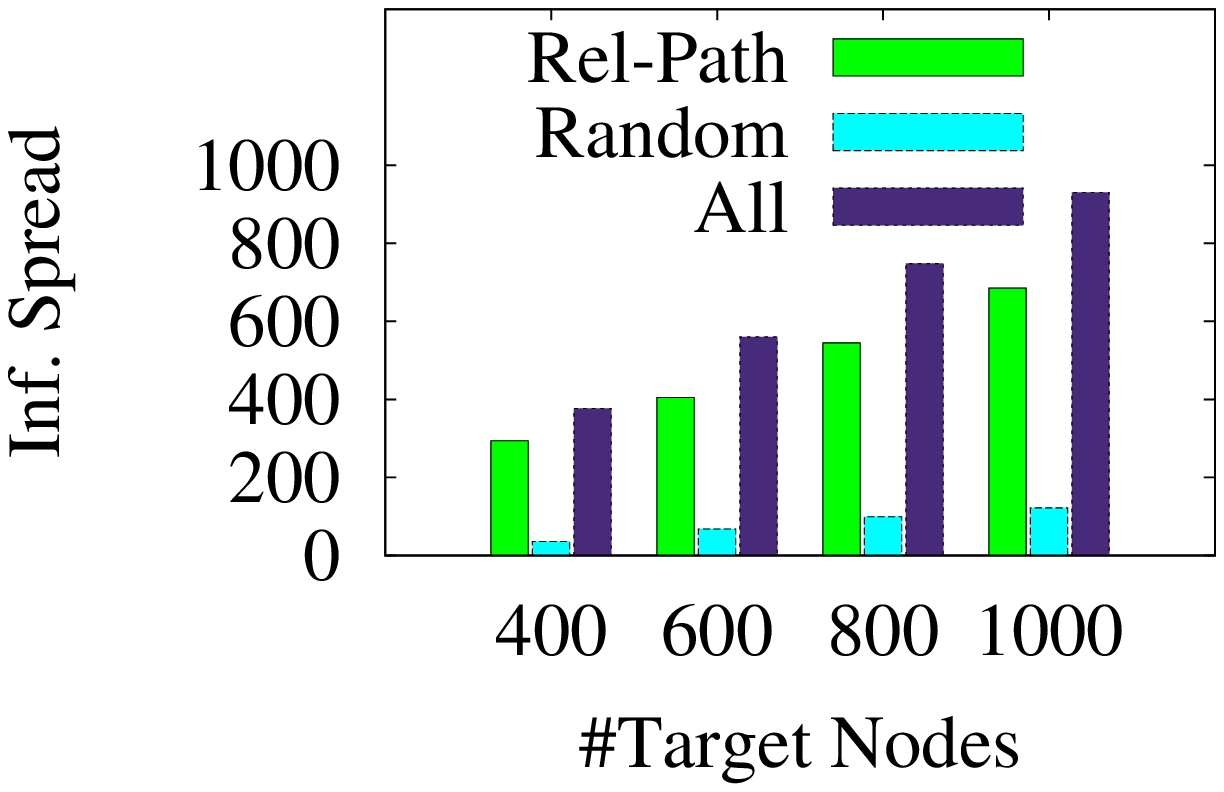}
\label{fig:inf_db_db}
}
\subfigure [\scriptsize Running Time] {
\includegraphics[scale=0.2]{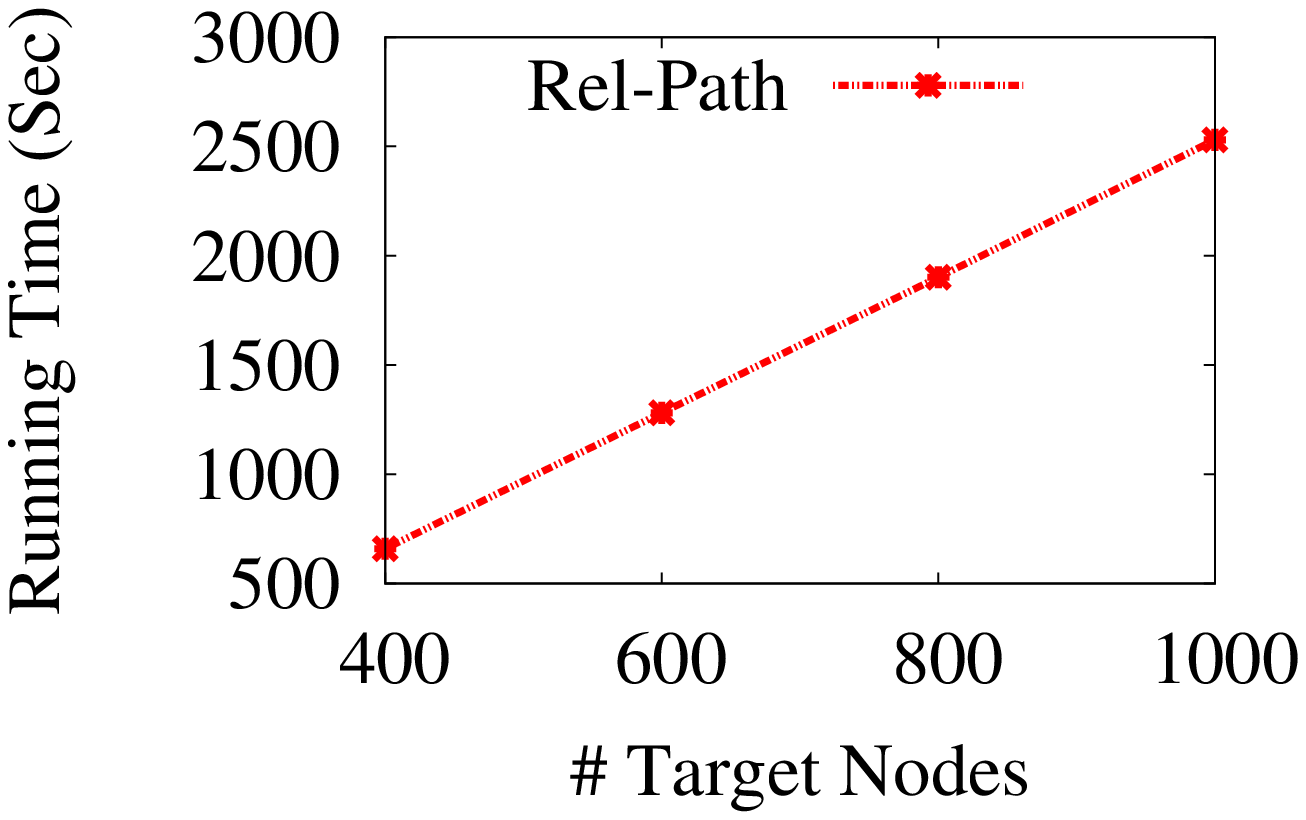}
\label{fig:time_inf_db_db}
}
\vspace{-4mm}
\caption{\scriptsize (a) Expected information spread by the top-$10$ catalysts and (b) running time to find the top-$10$ catalysts: {\em DBLP}, DB source nodes, DB target nodes}
\label{fig:db_db}
\vspace{-4mm}
\end{figure}
\begin{figure}[tb!]
\vspace{-1mm}
\centering
\subfigure [\scriptsize Information Spread]{
\includegraphics[scale=0.2]{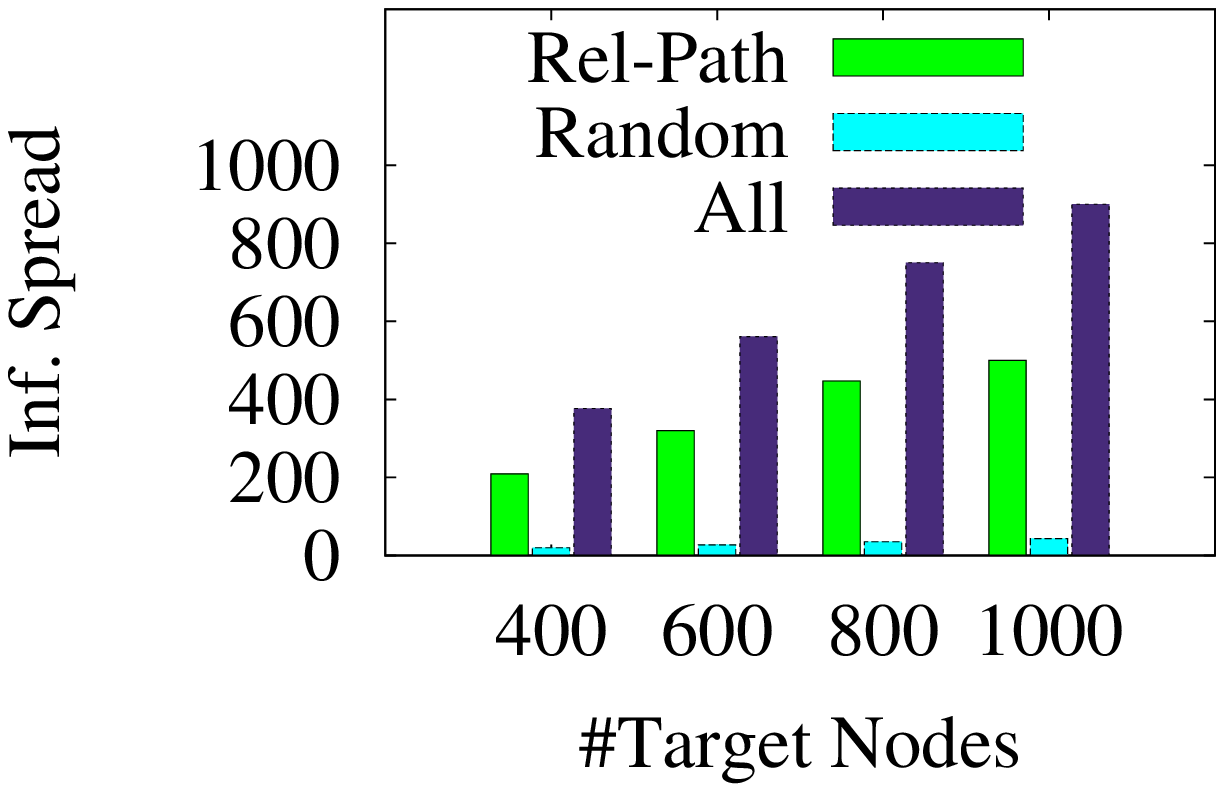}
\label{fig:inf_arch_db}
}
\subfigure [\scriptsize Running Time] {
\includegraphics[scale=0.2]{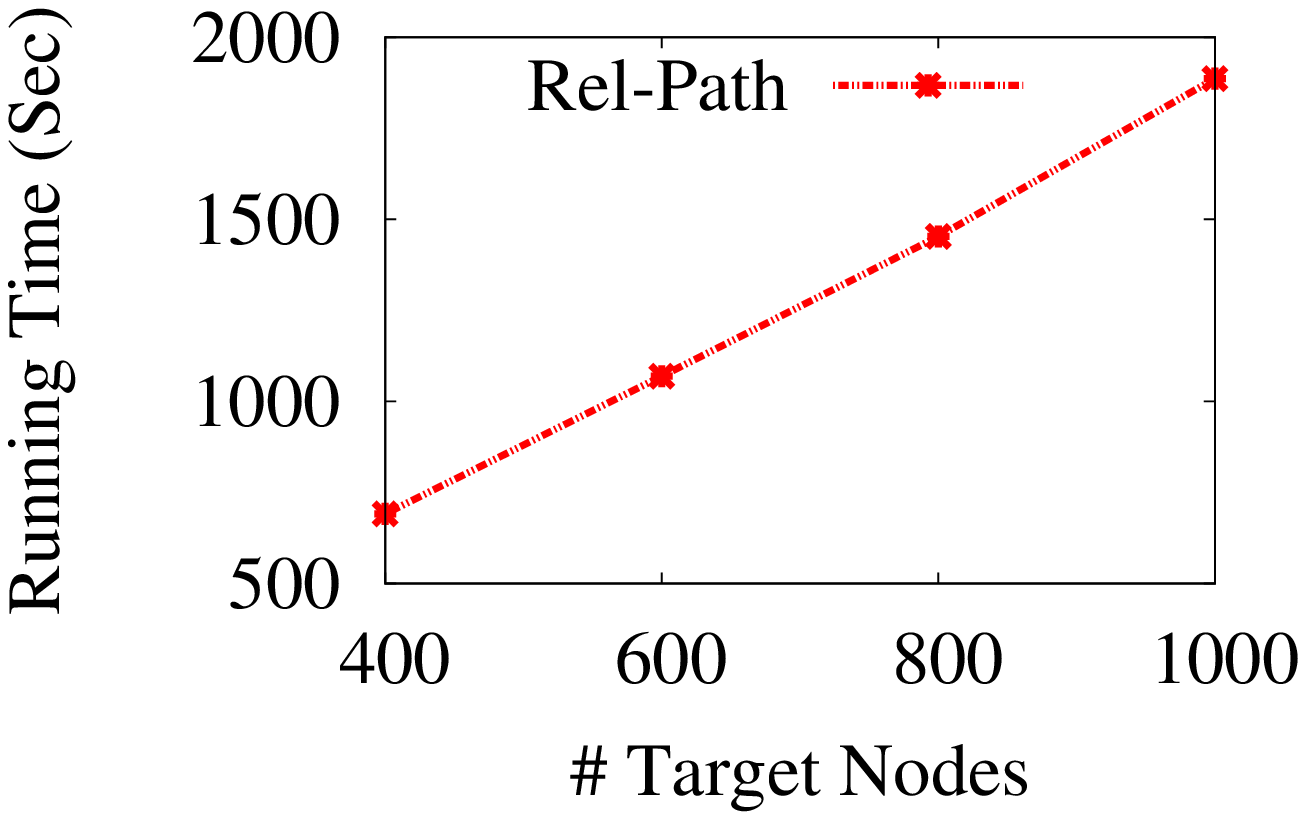}
\label{fig:time_inf_arch_db}
}
\vspace{-4mm}
\caption{\scriptsize (a) Expected information spread by the top-$10$ catalysts and (b) running time to find the top-$10$ catalysts: {\em DBLP}, ARCH source nodes, DB target nodes}
\label{fig:arch_db}
\vspace{-5mm}
\end{figure}

\noindent \underline{Source nodes from Databases.} We find the top-$10$ authors having the maximum number of publications in top-tier
database conferences and journals. They are: \{\textsf{Divesh Srivastava}, \textsf{Surajit Chaudhuri}, \textsf{Jiawei Han}, \textsf{Philip S. Yu}, \textsf{Hector Garcia-Molina},
\textsf{Jeffrey F. Naughton}, \textsf{H. V. Jagadish}, \textsf{Michael Stonebraker}, \textsf{Beng Chin Ooi}, \textsf{Raghu Ramakrishnan}\}.
\newline\underline{Source nodes from Computer Architecture.} In an analogous manner, we select the top-$10$ authors from the computer
architecture domain:  \{\textsf{Alberto L. Sangiovanni-Vincentelli}, \textsf{Jingsheng Jason Cong}, \textsf{Massoud Pedram},
\textsf{Andrew B. Kahng}, \textsf{Robert K. Brayton}, \textsf{Yao-Wen Chang}, \textsf{David Blaauw}, \textsf{Miodrag Potkonjak}, \textsf{Kaushik Roy}, \textsf{Xianlong Hong}\}.
\newline\underline{Target nodes from Databases.} We consider authors having at least $5$ publications in top-tier database conferences and journals
as our target nodes. We vary the number of target nodes from 400 to 1000, selected uniformly at random from them, to demonstrate the scalability of our algorithm.

In Figures~\ref{fig:inf_db_db} and \ref{fig:inf_arch_db}, we show the expected information spread achieved by the top-$10$ catalysts selected
via our \oursshort method, under two scenarios, respectively, {\bf Case-1:} both source nodes and target nodes are from databases (DB), {\bf Case-2:}
source nodes are from architecture (ARCH), and target nodes from databases (DB).
%We find that the \btwo baseline does not finish in one day for these experiments, as we are dealing with
%more source and target nodes.
%Hence,
To demonstrate the quality of our results, we
report the expected information spread achieved by uniformly at random selection of $10$ catalysts ({\em denoted as ``{\sf Random}'' in the figures}).
We observe from Figures~\ref{fig:inf_db_db} and \ref{fig:inf_arch_db} that \oursshort selects high-quality catalysts, and significantly
outperforms such a {\sf Random} method.

In particular, the catalysts selected by \oursshort under Case-1 are all DB-related, e.g., database systems, relational, information
extraction, keyword search, XML, data mining, etc. On the other hand, the catalysts selected by \oursshort under Case-2 belong to DB or
ARCH areas, e.g., CMOS, FPGA, storage system, cache, VLSI circuit, On-chip, transactional memory, data stream, etc. Both in
Figures~\ref{fig:inf_db_db} and \ref{fig:inf_arch_db}, we also report the total information spread achieved by all 5\,428 catalysts
(i.e., keywords) present in the {\em DBLP} dataset. This is denoted as ``{\sf All}'' in the figures. We find that {\em the information
spread achieved by only the top-10 catalysts is generally within 70-90\% of the total information spread achieved by all 5\,428
catalysts}. These results demonstrate the relevance of our novel problem and its solution in the domain of information cascade
over social influence networks.

Furthermore, we find that that running time to find the top-$k$ catalysts via \oursshort increases almost linearly with more
target nodes (see Figures~\ref{fig:time_inf_db_db} and \ref{fig:time_inf_arch_db}), which illustrates the {\em scalability} of our technique.  
\vspace{-4mm}
\section{Related Work}
\label{sec:related}
\vspace{-2mm}
To the best of our knowledge, the problem of finding the top-$k$ catalysts
for maximizing the conditional reliability, that we study in this work, is novel.
In the following, we provide an overview of relevant work in neighboring areas.

\vspace{-1mm}
\spara{Reliability queries in uncertain graphs.} Reliability is a classic problem studied in systems and device networks
\cite{AMG75}. Reliability has been recently studied in the context of large social and biological networks.  Due to its \sharpP-completeness \cite{B86},
efficient sampling, pruning, and indexing methods have been considered~\cite{JLDW11,PBGK10,LYMJ14,KBGG14,ZZL15}.
%Correlated edge probabilities have been studied in \cite{Cheng2014,PBGK10,YWCW12,TZC014}, that is, edges are usually not independent of each other, instead
%there are always some hidden correlations among the edges sharing the same node. Such correlations can be modeled with a conditional probability table, e.g.,
%$p(e_1|e_2, e_3, \ldots)$ denotes that edge $e_1$ is conditioned on the existence of other edges. Liu et. al. \cite{Liu2014} considered weight uncertainty in which the weights associated with
%an edge are stochastic in nature. None of those works, however, model edge-existence probabilities conditioned on external factors.

\vspace{-1mm}
\spara{Constrained reachability queries.}
Mendelson and Wood show that finding all simple paths in a (deterministic) graph matching a regular expression is \NP-hard \cite{MW95}.
There are some query languages which support regular expression queries only in some restricted form, e.g., {\sf GraphQL}, {\sf SoQL}, {\sf GLEEN}, {\sf XPATH}, and {\sf SPARQL}.
Fan et. al. \cite{FLMTW11} study a special case of regular expressions
that can be solved in quadratic time. Edge-label constrained reachability and distance queries have been studied in \cite{JHWRX10,BGGU14}.

Label-constrained reachability queries have been also considered in the context of uncertain graphs \cite{CGBY14}.
However, in that work the goal was to estimate the reliability between two nodes under the constraint that paths connecting the two nodes contain only some admissible labels.
Thus, the input graph still has fixed edge probabilities that do not vary based on external conditions.
As a result, label-constrained reachability differs from conditional reliability introduced in this work, and, more importantly, our problem of finding the top-$k$ external conditions is not addressed in those works.

\vspace{-1mm}
\spara{Explaining relationships among entities.} Several works aim at identifying the best subgraphs/paths to describe how some input entities are related~\cite{FSYB11,SG10,FMT04}.
Sun et. al. propose {\sf PathSIM} \cite{SHYYW11} to find entities that are connected by similar relationship patterns.
However, all these works consider deterministic graphs. The semantics behind the notion of connectivity in uncertain graphs is different.

\vspace{-1mm}
\spara{Uncertain graphs with correlated edge probabilities.}
Although the bulk of the literature on uncertain graphs assumes edges to be independent of one another \cite{Boldietal12,CGBY14,JLDW11,LYMJ14}, some works deal with correlated edge probabilities, where the existence of an edge may depend on the existence of other edges in the graph (typically, edges sharing an end node) \cite{Cheng2014,ChengThreshold2015,PBGK10,YWWC11}.
%,TZC014}.
Another model that differs from the classic one is the one adopted by Liu et. al. \cite{Liu2014}, which considers that every edge is assigned a (discrete) probability density function over a set of possible edge weights.
However, none of those works model edge-existence probabilities conditioned on external factors, nor they study the problem of finding the top-$k$ factors that maximize the reliability between two (sets of) nodes.

\vspace{-1mm}
\spara{Topic-aware influence maximization}.
The classical problem of influence maximization has been recently considered in a topic-aware fashion \cite{BBM12,CFLFTT15}.
Although the input to that problem is similar to the input considered in this work (an uncertain graph where edge probabilities depend on some conditions),
topic-aware influence maximization solves a different problem, i.e., finding a set of seed nodes that maximize the spread of information for a given topic set.
Topic-aware influence maximization can however benefit from the solutions provided by our top-$k$ catalysts problem, e.g., in the case where topics are not known in advance.
A recent work by Li et. al. \cite{LFZT17} focuses on the problem of finding a size-$k$ tag set that maximizes the {\em expected spread} of influence started from a given source node. Our work is different as we aim at finding the top-$k$ external factors maximizing the {\em reliability} between two given (sets of) nodes.

\vspace{-1mm}
\spara{Difference with our prior work}. A preliminary version of this work was published as a short paper in \cite{KGWB15}.
The present version contains a lot of new significant material: a complete piece of research work concerning the generalization to the case of multiple source/target nodes, including problem formulations, theory, applications, algorithms, and experiments; important theoretical findings; more details, examples, and motivations for all the proposed algorithms, including detailed time-complexity analyses; a lot of additional experiments, including applications in information cascade; a detailed overview of the related literature.
%For the single-source single-target top-k catalyst problem, we included important theoretical characterization, algorithm sketches, case studies, and experimental evaluation. We also introduced the multiple- source multiple-target top-$k$ catalyst problem under two scenarios: maximizing an aggregate function and maximizing connectivity, and empirically demonstrated the effectiveness and efficiency of our proposed algorithms.

%!TEX root = conditional_reliability.tex
\vspace{-5mm}
\section{Conclusions}
\label{sec:conc}

\vspace{-1mm}
We formulated and investigated a novel problem of identifying the top-$k$ catalysts that maximize the reliability between source and target nodes in an uncertain graph.
We proposed a method based on iterative reliable-path inclusion.
Our experiments show that the proposed method achieves better
quality and significantly higher efficiency compared to simpler baselines.
In future, we shall consider more complex relationships between an edge and the catalysts, and other problems from the perspective of top-$k$ catalysts, e.g., nearest neighbors and influence maximization.

%\vspace{-3mm}
%\section{Acknowledgement} Arijit Khan is supported by MOE Tier-1 M401020000 and NTU M4081678. Any opinions, findings, and conclusions in this publication are those of the authors and do not necessarily reflect the views of %the funding agencies.

\vspace{-3mm}
{\scriptsize
\bibliographystyle{abbrv}
\bibliography{ref,ref_rel}
}

\vspace{-6mm}
\appendix
\vspace{-2mm}
\spara{Proof of Theorem 2.}
A problem is said to admit a \emph{Polynomial Time Approximation Scheme} (\PTAS) if the problem admits a polynomial-time constant-factor approximation algorithm for \emph{every} constant $\beta \in (0,1)$.
We prove the theorem by showing that there exists at least one value of $\beta$ such that, if a $\beta$-approximation algorithm for \topcat exists, then we can solve the well-known \setcov problem in polynomial time. Since \setcov is an \NP-hard problem, clearly this can happen only if \Poly\ = \NP.

In \setcov we are given a universe $U$, and a set of $h$ subsets of $U$, i.e., $\mathcal{S} = \{S_1, S_2,\ldots, S_h\}$,
where $S_i \subseteq U$, for all $i \in [1\ldots h]$. The decision version of \setcov asks the following question: given $k$, is there
any a solution with no more than $k$ sets that cover the whole universe?

Given an instance of \setcov, we construct in polynomial time an instance
of our \topcat problem in the same way as in Theorem \ref{th:np_hard}.
On this instance, if $k$ sets suffice to cover the whole universe in the original instance of \setcov, the optimal solution $C^*$
 would have reliability at most $[1-(1-p^2)^Z]$, where $Z = |U|$ (because at most $Z$ disjoint paths from $s$ to $t$ would be produced, each with existence probability $p^2$).
On the other hand, if no $k$ sets cover the whole universe, $C^*$ would have reliability at most $[1-(1-p^2)^{Z-1}]$ (because at least one of the disjoint paths would be discarded).

Now, assume that a polynomial-time $\beta$-approximation algorithm for \topcat exists, for some $\beta \in (0,1)$. Call it  ``{\sf Approx}''.
{\sf Approx} would yield a solution $C_2$ such that
$R\left((s,t)|C_2\right) \ge \beta R\left((s,t)|C^*\right)$.
Now, consider the inequality $[1-(1-p^2)^{Z-1}] < \beta [1-(1-p^2)^Z]$.
If this inequality has solution for some values of $\beta$ and $p$, then
by simply running {\sf Approx} on the instance of \topcat constructed this way, and checking the reliability of the solution returned by {\sf Approx}, one can answer \setcov in polynomial time: a solution to \setcov exists iff the solution given by {\sf Approx} has reliability $\ge \beta[1-(1-p^2)^Z]$.
Thus, to prove the theorem we need to show that a solution to that inequality exists.

To this end, consider the real-valued function $f(p,Z) =\frac{1-(1-p^2)^{Z-1}}{1-(1-p^2)^Z}$.
Our inequality has a solution iff $\beta > f(p, Z)$.
It is easy to see that $f(p,Z) < 1$, for all $Z \ge 1$ and $p > 0$.
This means that there will always be a value of $\beta \in (0,1)$ and $p$ for which $\beta > f(p, Z)$ is satisfied, regardless of $Z$.
Hence, there exists at least one value of $\beta$ such that the inequality $[1-(1-p^2)^{Z-1}] < \beta [1-(1-p^2)^Z]$ has solution, and, based on the above argument, such that no $\beta$-approximation algorithm for Problem \ref{prob:topk_rel_col_set} can exist.
The theorem follows.

%\vspace{-2mm}
\spara{Proof of Theorem 4.}
If both our objective function and the constraint were proved to be submodular, our iterative path inclusion problem (Problem 2) would become an instance of the  {\em Sub-modular Cost Sub-modular Knapsack} ({\sf SCSK}) problem \cite{IB13}, and the approximation result in Theorem 4 would easily follow from \cite{IB13}.
%, where a similar performance guarantee is proved for a sub-modular cost sub-modular knapsack problem.
In the following we show that indeed both our objective function (Lemma 1) and our constraints (Lemma 2) are submodular, thus also proving Theorem 4.

\vspace{-1mm}
\begin{lemma}
The constraint of the iterative path inclusion problem, i.e., total number of catalysts on
edges of the included paths is sub-modular with respect to inclusion of paths.
\label{th:submodular_catalyst}
\end{lemma}

\vspace{-2mm}
\begin{proof}
Consider two path sets $\mathcal{P}_1$, $\mathcal{P}_2$ from $s$ to $t$ such that $\mathcal{P}_2 \supseteq\mathcal{P}_1$. Also,
we assume a path $P$ from $s$ to $t$, where $P \not \in \mathcal{P}_2$. There can be two distinct cases: (a) $P$ has no common catalyst
with the paths in $\mathcal{P}_2\setminus\mathcal{P}_1$. (b) $P$ has at least one common catalyst with the paths in
$\mathcal{P}_2\setminus\mathcal{P}_1$. In the first case,
{\tiny
\begin{align}
\displaystyle \left|\underset{e\in \mathcal{P}_1\cup\{P\}}{\cup} C(e)\right|-\left|\underset{e\in\mathcal{P}_1}{\cup}C(e)\right| = \left|\underset{e\in\mathcal{P}_2\cup\{P\}}{\cup}C(e)\right|-\left|\underset{e\in\mathcal{P}_2}{\cup}C(e)\right|
\end{align}
}%
In the second case,
{\tiny
\begin{align}
\displaystyle \left|\underset{e\in \mathcal{P}_1\cup\{P\}}{\cup} C(e)\right|-\left|\underset{e\in\mathcal{P}_1}{\cup}C(e)\right|  < \left|\underset{e\in\mathcal{P}_2\cup\{P\}}{\cup}C(e)\right|-\left|\underset{e\in\mathcal{P}_2}{\cup}C(e)\right|
\end{align}
}%
Hence, the result follows.
\end{proof}
\vspace{-3mm}
\begin{lemma}
If the top-$r$ most reliable paths are node-disjoint (except at source and target nodes), then the objective
function of the iterative path selection problem (Problem~\ref{prob:path}), i.e., $Rel_{\mathcal{P}_1}(s,t)$ is sub-modular with respect to inclusion of paths.
\label{th:submodular_nonoverlap}
\end{lemma}
\vspace{-3mm}
\begin{proof}
Assume $\mathcal{P}_1,\mathcal{P}_2 \subset \mathcal{P}$, such that $\mathcal{P}_1 \subseteq \mathcal{P}_2$. Also consider a path $P \in \mathcal{P}$
and $P \not \in \mathcal{P}_2$. Let us denote by $Rel_{\mathcal{P}_1}(s,t)=p_1$, $Rel_{\mathcal{P}_1\cup\{P\}}(s,t)=p_1+\delta$, and
$Rel_{\mathcal{P}_2\setminus\mathcal{P}_1}(s,t)=p_2$. Due to our assumption that the top-$r$ most reliable paths in $\mathcal{P}$ are node-disjoint except at the source and the target, we have: $Rel_{\mathcal{P}_2}(s,t)=1-(1-p_1)(1-p_2)$,
and $Rel_{\mathcal{P}_2\cup\{P\}}(s,t)=1-(1-p_1-\delta)(1-p_2)$. Hence, $Rel_{\mathcal{P}_2\cup\{P\}}(s,t)$ $-$ $Rel_{\mathcal{P}_2}(s,t)$
$=$ $(1-p_2)\delta$. This is smaller than or equal to $\delta$, which was the marginal gain for including the path $P$ in the set $\mathcal{P}_1$.
Therefore, our objective function is sub-modular.
\end{proof}

\IEEEoverridecommandlockouts
\vspace{-15mm}
\begin{IEEEbiography}[{\includegraphics[width=0.7in,height=0.9in]{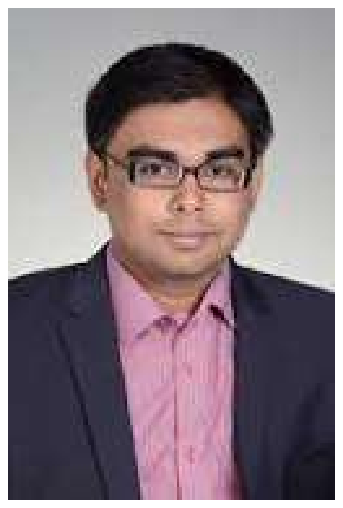}}]
{Arijit Khan} is an Assistant Professor at Nanyang
Technological University, Singapore. He earned his PhD from the University of
California, Santa Barbara, and did a post-doc in the Systems group at ETH Zurich.
Arijit is the recipient of the IBM PhD Fellowship in 2012-13.
He co-presented tutorials on graph queries and systems
at ICDE 2012, VLDB 2014, 2015, 2017.
\end{IEEEbiography}
\vspace{-17mm}
\begin{IEEEbiography}[{\includegraphics[width=0.7in,height=0.9in]{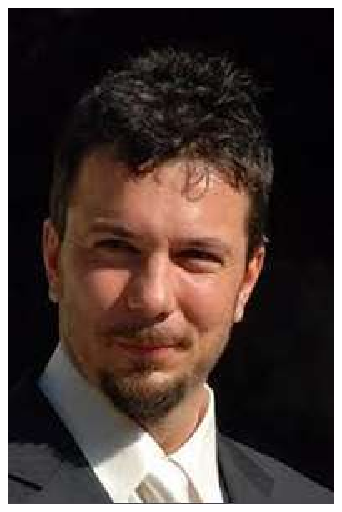}}]
{Francesco Bonchi} is a Research Leader at the ISI Foundation, Turin, Italy.
Before he was Director of Research at Yahoo Labs in Barcelona, Spain.
He is an Associate Editor of the IEEE Transactions on Knowledge and Data Engineering (TKDE), and the ACM Transactions on Intelligent Systems and Technology (TIST).
Dr. Bonchi has served as program co-chair of HT 2017, ICDM 2016, and ECML PKDD 2010.
%He earned his Ph.D. from the University of Pisa in 2003.
\end{IEEEbiography}
\vspace{-16mm}
\begin{IEEEbiography} [{\includegraphics[width=0.7in,height=0.9in]{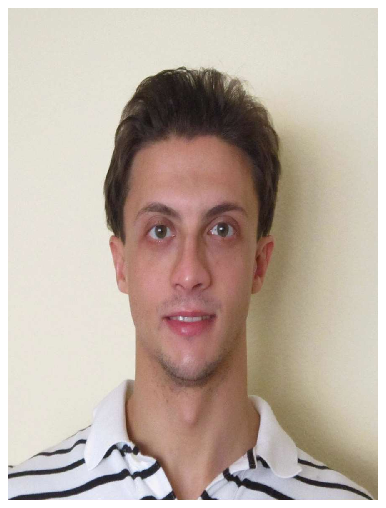}}]
{Francesco Gullo} is a research scientist at UniCredit,, R\&D department.
He received his Ph.D. from the University of Calabria (Italy) in 2010.
He spent four years at Yahoo Labs, Barcelona, first as a postdoctoral researcher, and then
as a research scientist. He served as workshop chair of ICDM'16, and organized several workshops/symposia (MIDAS @ECML-PKDD'16, MultiClust @SDM'14,
@KDD'13).
\end{IEEEbiography}
\vspace{-16mm}
\begin{IEEEbiography}[{\includegraphics[width=0.7in,height=0.9in]{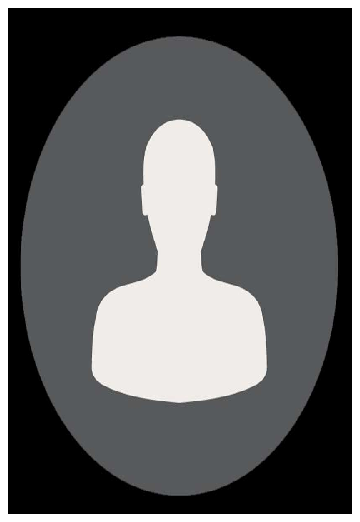}}]
{Andreas Nufer} is a Senior Consultant at GridSoft AG in Switzerland. He completed his masters in Computer Science at ETH Zurich,
and did his bachelors at Ecole Superieure en Sciences Informatiques in France, and also from Bern University of Applied Science
in Switzerland.
\end{IEEEbiography}

\end{document}